\documentclass[11pt]{article}

\usepackage{fullpage}

\usepackage{xspace}
\usepackage{amsfonts,amsmath,amssymb, amsthm}
\usepackage{cite}
 \usepackage[backref,colorlinks,citecolor=blue,bookmarks=true,linktocpage]{hyperref}
 \usepackage{aliascnt}

\makeatletter
\@ifundefined{theorem}{    \theoremstyle{plain}
  	\newtheorem{theorem}{Theorem}[section]
  	\newaliascnt{corollary}{theorem}
  	  
  	\aliascntresetthe{corollary}
  	\newaliascnt{lemma}{theorem}
  		\newtheorem{lemma}[lemma]{Lemma}
  	\aliascntresetthe{lemma}
  	\newaliascnt{claim}{theorem}
  		
	\aliascntresetthe{claim}
	\newaliascnt{fact}{theorem}
 	 	\newtheorem{fact}[theorem]{Fact}
	\aliascntresetthe{fact}
  	
  	\newaliascnt{proposition}{theorem}
  		\newtheorem{proposition}[proposition]{Proposition}
	\aliascntresetthe{proposition}
  \theoremstyle{remark}
  	\newtheorem{remark}[theorem]{Remark}

  \theoremstyle{definition}
  	\newaliascnt{definition}{theorem}
 		 \newtheorem{definition}[definition]{Definition}
 	 \aliascntresetthe{definition}
}{}
\makeatother

\title{Communication with Imperfectly Shared Randomness\thanks{An extended abstract of this work appeared in the \emph{Proceedings of the 2015 Conference on Innovations in Theoretical Computer Science (ITCS)} \cite{CGMS:15}.}}
\date{\today}
\author{Cl\'ement L. Canonne\thanks{Columbia University. Email: \email{ccanonne@cs.columbia.edu}. Research supported in part by NSF CCF-1115703 and NSF CCF-1319788. Some of this work was done when the author was an intern at Microsoft Research New England.},
Venkatesan Guruswami\thanks{Carnegie Mellon University. Email: \email{venkatg@cs.cmu.edu}.
Some of this work was done when the author was visiting Microsoft Research New England. Research supported in part by NSF CCF-0963975.
},
Raghu Meka\thanks{University of California, Los Angeles. Email: \email{raghum@cs.ucla.edu} .},
Madhu Sudan\thanks{Harvard University. Email: \email{madhu@cs.harvard.edu}.}}

\usepackage[numbered]{bookmark}
\usepackage{mleftright}

\newcommand{\appendixref}[1]{\hyperref[#1]{Appendix~\ref{#1}}}

\usepackage[shortlabels]{enumitem}

\usepackage{tikz}
\usetikzlibrary{arrows,shapes}

\def\withcolors{0}

\usepackage{thm-restate}

\newenvironment{proofof}[1]{\begin{proof}[Proof of {#1}]}{\end{proof}}

\providecommand{\email}[1]{\href{mailto:#1}{\nolinkurl{#1}\xspace}}

\ifnum\withcolors=1
  
  \newcommand{\newer}[1]{{\color{blue} {#1}}}
  \newcommand{\newest}[1]{{\color{orange} {#1}}}
  
  \newcommand{\mnew}[1]{{\color{magenta} {#1}}}
\else
  
  \newcommand{\newer}[1]{{{#1}}}
  \newcommand{\newest}[1]{{{#1}}}
  
  \newcommand{\mnew}[1]{{{#1}}}
\fi

\newcommand{\ignore}[1]{}

\newcommand{\eps}{\ensuremath{\epsilon}\xspace}
\newcommand{\eqdef}{\stackrel{\rm def}{=}}
\newcommand{\yes}{\textsf{yes}\xspace}
\newcommand{\no}{\textsf{no}\xspace}
\newcommand{\dyes}{{\cal Y}}
\newcommand{\dno}{{\cal N}}

\newcommand{\bigO}[1]{{O\!\left({#1}\right)}}
\newcommand{\bigTheta}[1]{{\Theta\!\left({#1}\right)}}
\newcommand{\bigOmega}[1]{{\Omega\!\left({#1}\right)}}

\providecommand{\poly}{\operatorname*{poly}}

\newcommand{\infl}[2][f]{{\mathsf{Inf}_{#1}(#2)}}
\newcommand{\infldeg}[3][f]{{\mathsf{Inf}_{#1}^{#2}(#3)}}
\newcommand{\Var}{\mathsf{Var}}
\newcommand{\multindex}{\boldsymbol{\sigma}}

\newcommand{\setOfSuchThat}[2]{ \left\{\; #1 \;\colon\; #2\; \right\} } 			\newcommand{\dist}[2]{{\operatorname{dist}\!\left({#1, #2}\right)}} 
\newcommand{\proba}{\Pr}
\newcommand{\probaOf}[1]{\proba\!\left[\, #1\, \right]}
\newcommand{\probaCond}[2]{\proba\!\left[\, #1 \;\middle\vert\; #2\, \right]}
\newcommand{\probaDistrOf}[2]{\proba_{#1}\left[\, #2\, \right]}

\newcommand{\expect}[1]{\mathbb{E}\!\left[#1\right]}

\newcommand{\shortexpect}{\mathbb{E}}

\newcommand{\uniform}{\ensuremath{\mathcal{U}}}
\newcommand{\bernoulli}[1]{\ensuremath{\operatorname{Bern}\!\left( #1 \right)}}
\newcommand{\bern}[2]{\ensuremath{\operatorname{Bern}^{#1}\!\left( #2 \right)}}

\newcommand{\norm}[1]{\lVert#1{\rVert}}

\newcommand{\normtwo}[1]{{\norm{#1}}_2}

\newcommand{\abs}[1]{\left\lvert #1 \right\rvert}

\newcommand{\dotprod}[2]{ \left\langle #1,\xspace #2 \right\rangle } 			\newcommand{\ip}[2]{\dotprod{#1}{#2}}

\newcommand{\R}{\ensuremath{\mathbb{R}}\xspace}
\newcommand{\Z}{\ensuremath{\mathbb{Z}}\xspace}
\newcommand{\N}{\Z^+}
\newcommand{\wt}{\mathrm{wt}}
\newcommand{\round}{\mathsf{Round}}
\newcommand{\mv}[1]{\mathbf{#1}}
\newcommand{\good}{\textsc{Good}\xspace}

\DeclareMathOperator{\cc}{cc}
\DeclareMathOperator{\isrcc}{isr-cc}
\DeclareMathOperator{\psrcc}{psr-cc}
\DeclareMathOperator{\privatecc}{private-cc}
\DeclareMathOperator{\capsisrcc}{\mathsf{ISR-CC}}
\DeclareMathOperator{\capspsrcc}{\mathsf{PSR-CC}}

\newcommand{\owisrcc}[1][]{\isrcc_{#1}^{\rm ow}}
\DeclareMathOperator{\owpsrcc}{\psrcc^{\rm ow}}

\newcommand{\smisrcc}[1][]{\isrcc_{#1}^{\rm sm}}

\newcommand{\capsowisrcc}[1][]{\capsisrcc_{#1}^{\rm ow}}

\newcommand{\capssmisrcc}[1][]{\capsisrcc_{#1}^{\rm sm}}

\DeclareMathOperator{\success}{succ}

\newcommand{\calf}{\mathcal{F}}
\newcommand{\calg}{\mathcal{G}}
\newcommand{\bad}{\textsc{Bad}}
\newcommand{\ext}{\mathsf{Ext}}

\newcommand{\compress}{\textsc{Compress}\xspace}
\newcommand{\agreement}{\textsc{Agreement-Distillation}\xspace}
\newcommand{\gapip}{\textsc{GapInnerProduct}\xspace}
\newcommand{\sparsegapip}{\textsc{SparseGapInnerProduct}\xspace}
\newcommand{\gausscorr}{\textsc{GaussianCorrelation}\xspace}

\newcommand{\strat}[1]{\xi_{#1}}
\newcommand{\tildestrat}[1]{\tilde{\xi}_{#1}}
\newcommand{\ptranscript}[1]{\chi_{#1}}
\newcommand{\tildeptranscript}[1]{\tilde{\chi}_{#1}}

\newcommand{\E}{\shortexpect} 
\begin{document}
\maketitle

\begin{abstract}
Communication complexity investigates the amount of communication needed for two or more players to determine some joint function of their private inputs. For many interesting functions, the communication complexity can be much smaller than basic information theoretic measures associated with the players' inputs such as the input length, the entropy, or even the conditional entropy. Communication complexity of many functions reduces further when the players share randomness.  Classical works studied the communication complexity of functions when the interacting players share randomness perfectly, i.e., they get identical copies of randomness from a common source. This work considers the variant of this question when the players share randomness imperfectly, i.e., when they get noisy copies of the randomness produced by some common source. Our main result shows that any function that can be computed by a $k$-bit protocol in the perfect sharing model has a $2^k$-bit protocol in the setting of imperfectly shared randomness and such an exponential growth is necessary. Our upper bound relies on ideas from locality sensitive hashing while lower bounds rely on hypercontractivity and a new invariance principle tailored for communication protocols.
\end{abstract}
\setcounter{page}{0}
\pagenumbering{gobble}
\clearpage
\tableofcontents
\newpage
\pagenumbering{arabic}
\section{Introduction}
\label{sec:intro}

The availability of shared randomness can lead to enormous savings in communication complexity
when computing some basic functions whose inputs are spread out over
different communicating players. A basic example of this is Equality
Testing, where two players Alice and Bob have inputs $x\in\{0,1\}^n$
and $y\in\{0,1\}^n$ and need to determine if $x = y$. Deterministically this
takes $n$ bits of communication. This reduces to $\Theta(\log n)$ bits if Alice and
Bob can toss coins and they are allowed some error. But if they share some
randomness $r \in \{0,1\}^*$ independent of $x$ and $y$ then the
communication cost drops to $O(1)$. (See, for instance, \cite{KusNisan:06:book}.)

A more prevalent example of a communication problem is compression with
uncertain priors. Here Alice has a distribution $P$ on a universe $[N] =
\{1,\ldots,N\}$, and a message $m \in [N]$ chosen according to the
distribution $P$. Alice is allowed to send some bits to Bob and Bob should
output $m$ and the goal is to minimize the expected number of bits that
Alice sends Bob (over the random choice of $m$). If Bob knows the
distribution $P$ exactly then this is the classical compression problem,
solved for example by Huffman coding. In most forms of natural
communication (e.g., think about the next email you are about to send),
Alice and Bob are not perfectly aware of the underlying context to their
exchange, but have reasonably good ideas about each other. One way to model
this is to say that Bob has a distribution $Q$ that is \emph{close} to
the distribution $P$ that Alice is working with, but is not identical to
$P$. Compressing information down to its entropy in the presence of such
uncertainty (i.e., $P \neq Q$) turns out to be possible if Alice and Bob
share randomness that is independent of $(P,Q,m)$ as shown by Juba et al.~\cite{JKKS2011}.
However it remains open as to whether such compression
can be effected deterministically, without the shared randomness --- the
best known schemes can only achieve a compression length of roughly $O(H(P) + \log
\log N)$, where $H(P) = \sum_{i\in[N]} P(i) \log 1/P(i)$ denotes the entropy
of $P$.\footnote{We stress that the setting of uncertain compression is completely different from that of compression with the ``wrong distribution'', a well-studied question in information theory. In the ``wrong distribution problem'' (see, for instance, \cite[Theorem 5.4.3]{CoverJoy:book}) the sender and receiver agree on the distribution, say $P$, but both have it wrong and the distribution the message comes from is $R$. This leads to a compression length of $\E_{m \sim R}[\log (1/P(m))] \approx H(R) + D(R \|P)$. The important aspect here is
that while the compression is not as good, there is no confusion between sender and receiver; and the latter is the focus of our problem.}

In both examples above it is natural to ask the question: can the
(presumed) savings in communication be achieved in the absence of perfect
sharing of randomness? The question especially makes sense in the latter
context where the essential motivation is that Alice and Bob are not in
perfect synchrony with each other: If Alice and Bob are not perfectly aware of the distributions $P$ and $Q$, why should their randomness be identical?

The question of communication with imperfectly shared
randomness was considered recently in the work of Bavarian et
al.~\cite{BavarianGI:14}. They consider the setting where Alice and Bob have
randomness $r$ and $s$ respectively, with some known correlation between
$r$ and $s$, and study the implications of {imperfectly shared} randomness in the simultaneous message communication model (where a referee gets messages from Alice and Bob and computes some joint function of their inputs). Their technical focus is on the different kinds of correlations possible between $r$ and $s$, but among basic results
they show that equality testing has a $O(1)$ communication complexity
protocol with {imperfectly} shared randomness.
{Similar questions have been considered in other contexts and communities, such as information theory~\cite{GacsKorner:73,Wits:75,AC:98,KA:12,FRKT:15,BG:15}, cryptography~\cite{BrassardSalvail:93,Maurer:93,AC:93,CN:00,RennerWolf:05}, probability theory~\cite{Mossel-cosmic:05,Mossel-NICD:06,BogdanovMossel,CMN:14}, and quantum computing~\cite{BBPSSW:96}.}

{In this work we are concerned with the setting of general communication protocols, where Alice and Bob
interact to determine the value of some function.}
From some perspectives, this setting does not seem to offer a major
difference between ``private randomness'' and ``perfectly shared
randomness'' --- Newman~\cite{Newman:91} shows that the communication
complexity in the former setting can be larger by at most an additive $\log
n$ term, where $n$ is the input size {(indeed, Newman proves that any protocol with perfectly shared
randomness can be converted into one that uses only $O(\log n)$ bits of shared randomness)}.
``Imperfectly shared randomness'' being in between the two models cannot therefore
be too far from them either. However, problems like compression above
highlight a different perspective. There $N$ is the size of the universe of
all possible messages, and compression to $\log N$ bits of
communication is trivial and uninteresting. Even a solution with $\log \log
N$ bits of communication is not completely satisfactory. The real
target is $O(H(P))$ bits of communication, which may be a constant
independent of the universe size $N$ (and for natural communication, the set
of possible messages could be thought of as an infinitely large set).
Thus the gap between the communication complexity with perfectly shared
randomness and imperfectly shared randomness remains a very interesting
question, which we explore in this paper.

We provide a formal description of our models and results in the following
section, and here give an informal preview. We consider communication
complexity in a simplified setting of imperfectly shared randomness: Alice
has a uniform binary string $r \in \{0,1\}^m$ and Bob has a string $s$
obtained by flipping each bit of $r$ independently with some tiny
probability. (While this setting is not the most general possible, it seems
to capture the most interesting aspects of the ``lack of prior
agreement'' between Alice and Bob.) Our main contributions in this work are
the introduction of some new problems of interest in the context of
communication complexity, and a comparison of their communication complexity
with/without perfect sharing of randomness.

The first problem we study is the complexity of {\em compression with
uncertain priors}.
We show that any distribution $P$ can be compressed to $O(H(P))$ bits
even when the randomness is not perfectly shared. As in the analogous
result of Juba et al.~\cite{JKKS2011} this protocol sheds some light on
natural communication processes, and introduces an error-correcting element
that was not previously explained.

The next problem we mention is that of {\em agreement distillation}.
Here Alice and Bob try to agree on a small random string using little communication. This is a natural problem to study in the context of communication
complexity with imperfect randomness, since an efficient solution for this problem would allow Alice and Bob to
convert any protocol using perfectly shared randomness into one that relies
only on imperfectly shared randomness. 
It turns out that the zero-communication version of this question, where Alice and Bob are not allowed to communicate at all with each other, ant the one-way communication version were studied in the past.
 Witsenhausen~\cite{Wits:75} shows for instance that no perfect
agreement is possible even for a single bit, i.e. Alice and Bob must fail with positive probability. 
Ahlswede and Csiszar~\cite{AC:98} studies the one-way communication version of this question and gives tight bounds on the number of bits that need to be communicated to get $k$ bits of entropy with probability  tending to one. 
Later, Bogdanov and
Mossel~\cite{BogdanovMossel} extend this negative result, showing that the probability that Alice and Bob can agree on a
$k$-bit string is exponentially small in $k$. By a reduction we show
that this implies that $o(k)$ bits of communication are insufficient to get
agreement on $k$ bits. Conversely, we also show that
Alice and Bob can get a
constant factor advantage --- so they can communicate $\alpha k$ bits for some
$\alpha < 1$ {to obtain $k$ bits of perfectly shared randomness {with high probability}}. Such a result seems implicit in \cite{BogdanovMossel}.

\newest{A sequence of earlier works~\cite{GacsKorner:73,Wits:75,AC:98} also studied the agreement distillation problem focusing on the maximum achievable ratio $a/c$, such that for sufficiently large number $r$ of used correlated samples, Alice and Bob can agree on $a \cdot r$ random bits using $c \cdot r$ bits of communication. Moreover, these works focus on the case where the agreement probability tends to 1 (as $r \to \infty$). It is
surprising that despite requiring the number of agreed bits to grow linearly in the number of used samples, Ahlswede and Csiszar~\cite{AC:98} lose nothing in terms of the best achievable trade-off.\

Following our work, Guruswami and Radhakrishnan~\cite{GR:16} pinpoint the exact trade-off between communication and success probability required in order for Alice and Bob to agree on $k$ bits of common randomness, when an unlimited number of correlated samples are available.}

Returning to our work,
we next attempt to get a general conversion of communication
protocols from the perfectly-shared setting to the imperfectly-shared
setting. We introduce a complete promise problem {\gapip} which
captures {two}-way communication, and use it to show that any problem with {a protocol using} $k$
bits of communication with perfectly shared randomness
also has a $\min\{\exp(k),k+\log n\}$ bit {(one-way)} protocol with imperfectly shared
randomness. While the protocol is simple, we feel its existence is somewhat
surprising; and indeed it yields a very different protocol for equality testing when compared with Bavarian et
al.~\cite{BavarianGI:14}.

Lastly, our {\em main technical result} is a matching lower bound giving a
parameterized family of promise problems, \sparsegapip, where the $k$'th problem can
be solved with $k$ bits of communication with perfect randomness, but
requires $\exp(\Omega(k))$ bits with imperfect sharing. This result builds a new
connection between influence of variables and communication complexity,
which may be of independent interest. Finally we conclude with a variety of
open questions.

\section{Model, Formal Description of Results and Main Ideas}
\label{sec:model}

Throughout the paper, we denote by $\N$ the set of positive integers, and by $[n]$ the set $\{1,\dots, n\}$. Unless specified otherwise, all logarithms are in base 2. We also recall, for $x\in[0,1]$, the definition of the binary entropy function $h(x)=-x\log x-(1-x)\log(1-x)$; furthermore, for any  $p\in[0,1]$, we will write $\bernoulli{p}$ for the Bernoulli distribution on {$\{0,1\}$} with parameter $p$, and $\bern{n}{p}$ for the product distribution on {$\{0,1\}^n$} of $n$ independent Bernoulli random
variables.
 For a distribution $P$ over a domain $\Omega$, we write $H(P)= \sum_{x\in \Omega} P(x) \log (1/P(x))$ for its entropy, and $x\sim P$ to indicate that $x$ is drawn from $P$. $\uniform_{\Omega}$ denotes the uniform distribution over $\Omega$.

Finally, for two elements $x,y\in\{+1,-1\}^n$, their \emph{Hamming distance} $\dist{x}{y}$ is defined as the number of coordinates in which they differ (and similarly for $x,y\in\{0,1\}^n$).

\subsection{Model}

We use the familiar model of communication complexity, augmented by the
notion of {imperfectly} shared randomness.
Recall that in the standard model, two players, Alice and Bob, have access
to inputs $x$ and $y$ respectively. A protocol $\Pi$ specifies the
interaction between Alice and Bob (who speaks when and what), and concludes
with Alice and Bob producing outputs $w_A$ and $w_B$ respectively.
A communication problem $P$ is (informally) specified by conditions on the
inputs and outputs $(x,y,w_A,w_B)$. In usual (promise) problems this is simply a
relationship on the 4-tuple. In sampling problems, this may be given by
requirements on the distribution of this output given $x$ and $y$.
For functional problems, $P = (f_A,f_B)$ and
the conditions require that $w_A = f_A(x,y)$ and $w_B = f_B(x,y)$.
A randomized
protocol is said to solve a functional problem $P$ if the outputs are correct with
probability at least $2/3$. The (worst-case) complexity of a protocol $\Pi$, denoted
$\cc(\Pi)$ is the maximum over all $x,y$ of the expected number of bits
communicated by $\Pi$. This is the main complexity measure of interest to us, although
distributional complexity will also be considered, as also any mix. (For instance, the
most natural measure in compression is a max-average measure.)

We will be considering the setting where Alice and Bob have access to an
arbitrarily long sequence of correlated random bits. For this definition it
will be convenient to let a random bit be an element of $\{+1,-1\}$. For $\rho\in[{0,1}]$,\footnote{{The definition extends to $\rho\in[-1,+1]$, but in this work we shall without loss of generality only be concerned with non-negative correlations.}} we
say a pair of bits $(a,b)$ are {\em $\rho$-correlated (uniform) bits} if
$\E[a] = \E[b] = 0$ and $\E[ab] = \rho$. We will consider the performance of
protocols when given access to sequences $(r,r^\prime)$ where each coordinate
pair $(r^{}_i,r^\prime_i)$ are $\rho$-correlated uniform bits chosen independently
for each $i$. We shall  write $r\sim_\rho r^\prime$ for such $\rho$-correlated pairs.

The {\em communication complexity of a problem $P$} with access to
$\rho$-correlated bits, denoted\footnote{All throughout ``isr'' stands for \emph{{imperfectly} shared randomness, while \emph{psr} refers to \emph{{perfectly}} shared randomness.}} $\isrcc_\rho(P)$ is the minimum over all
protocols $\Pi$ that solve $P$ with access to $\rho$-correlated bits of
$\cc(\Pi)$.
For integer $k$, we let $\capsisrcc_\rho(k)$ denote the collections
of problems $P$ with $\isrcc_\rho(P) \leq k$.
The one-way communication complexity and simultaneous message complexities\footnote{{Recall that the \emph{simultaneous message passing model} (SMP)~\cite{BabaiK:97} is defined as a communication game between 3 players: Alice, Bob, and a Referee. Given a function $f$ known to all players, Alice and Bob both receive inputs respectively $x$ and $y$, and send messages to the Referee who must compute the value $f(x,y)$.}}
are defined similarly (by restricting to appropriate protocols)
and denoted $\owisrcc[\rho](P)$ and $\smisrcc[\rho](P)$ respectively.
The corresponding complexity classes are denoted similarly by
$\capsowisrcc[\rho](k)$ and $\capssmisrcc[\rho](k)$.

Note that when $\rho = 1$ we get the standard model of communication with
shared randomness. We denote this measure by $\psrcc(P) = \isrcc_1(P)$, and write $\capspsrcc(k)$ for the corresponding complexity class.
Similarly, when $\rho = 0$ we get communication complexity with private
randomness $\privatecc(P) = \isrcc_0(P)$. We note that $\isrcc_\rho(P)$ is
non-increasing in $\rho$. Combined with Newman's Theorem~\cite{Newman:91}, we obtain:

\begin{proposition}
For every problem $P$ with inputs $x,y \in \{0,1\}^n$ and $0 \leq \rho \leq \rho' \leq 1$ we have
\[
\psrcc(P) \leq \isrcc_{\rho'}(P) \leq \isrcc_{\rho}(P) \leq \privatecc(P)
\leq \psrcc(P) + O(\log n).
\]
\end{proposition}
\noindent The proposition also holds for one-way  {communication, and (except for the last inequality) simultaneous messages.}

\subsection{Problems, Results and Techniques}

We now define some of the new problems we consider in this work and describe our main results.

\subsubsection{Compression}

\begin{definition}[Uncertain Compression]
For $\delta > 0$, $\Delta\geq 0$ and {integers $\ell,n$}, the \emph{uncertain compression
problem $\compress_{\Delta,\delta}^{\ell,n}$} is a promise
problem with Alice getting as input the pair $(P,m)$, where
$P = (P_1,\ldots,P_n)$ is a probability distribution on $[n]$ and
$m \in [n]$. Bob gets a probability distribution $Q$ on $[n]$.
The promises are that $H(P) \leq \ell$
and for every $i \in [n]$, $\abs{\log (P_i/Q_i)} \leq \Delta$.
The goal is for Bob to output $m$, i.e., $w_B = m$ with probability at least $1-\delta$.
The measure of interest here is the maximum, over $(P,Q)$ satisfying the promise, of the
expected one-way communication complexity when $m$ is sampled according to $P$.
\end{definition}

When $\Delta = 0$, this is the classical compression problem and Huffman
coding achieves a compression length of at most $\ell +1$; and this is optimal
for ``prefix-free'' compressions. For larger values of $\Delta$, the work of
\cite{JKKS2011} gives an upper bound of $\ell + 2\Delta + O(1)$ in the setting of perfectly shared randomness (to get constant error probability).
In the setting of deterministic communication or private randomness,
it is open if this communication complexity can be bounded by a function of $\ell$ and
$\Delta$ alone (without dependence on $n$).
(The work of \cite{HaramatyS:2014} studies the deterministic
setting.) Our first result shows that the bound of \cite{JKKS2011} can be extended
naturally to the setting of imperfectly shared randomness.

\begin{restatable}{theorem}{thmcompress}\label{thm:compress}
For every $\eps, \delta > 0$ and $0 < \rho \leq 1$ there exists
$c = c_{\eps,\delta,\rho}$ such
that for every $\ell,n$, we have
$\owisrcc[\rho]\mleft(\compress_{\Delta,\delta}^{\ell,n}\mright)
\leq \frac{1+\eps}{1 - h((1-\rho)/2)}\mleft({H(P) + 2\Delta + c}\mright)$.
\end{restatable}

We stress that the notation
$\owisrcc[\rho]\left(\compress_{\Delta,\delta}^{\ell,n}\right)$
describes the {\em worst-case} complexity over $P$ with entropy $H(P) \leq \ell$
of the {\em expected} compression length when $m \sim P$. 
{We first note that one approach would be to initially ``distill'' perfectly shared randomness from the imperfectly shared one available (by communicating a few bits), before using this perfectly shared randomness to run the protocol of~\cite{JKKS2011}. Unfortunately, this results in $\bigTheta{\log n}$ bites of communication, which would be excessively large. Indeed, a naive protocol that ignores $P$ would only require to communicate $\log n$ bits. Instead, to achieve our bound we develop a new protocol based on} a simple modification of the protocol of~\cite{JKKS2011}. Roughly, Alice and Bob use their {imperfectly shared} randomness to define a
``redundant and ambiguous dictionary'' with words of every length for every message.
Alice communicates using a word of appropriate length given the distribution $P$,
and Bob decodes using maximum likelihood decoding given $Q$. The main difference in our
case is that Alice and Bob work knowing their dictionaries do not match exactly (as if they spelled the same words
differently) and so use even longer words during encoding and decoding with some
error-correction to allow for spelling errors. Details can be found in~\autoref{sec:compress}.  
\subsubsection{Agreement distillation}

Next we turn to a very natural problem in the context of imperfect sharing of
randomness. Can Alice and Bob communicate to distill a few random bits from their large
collection $r$ and $r^\prime$ (of correlated random bits), bits on which they can agree perfectly?

\begin{definition}[Agreement distillation]
In the $\agreement^k_\gamma$ problem, Alice and Bob have no inputs. Their goal is
to output $w_A$ and $w_B$ satisfying the following properties:
\begin{enumerate}[(i)]
	\item $\probaOf{w_A=w_B}\geq \gamma$;
	\item $H_\infty(w_A) \geq k$; and
	\item $H_\infty(w_B) \geq k$
\end{enumerate}
where $H_\infty(X)=\min_{x}\log\frac{1}{\probaOf{X=x}}$ {denotes the \emph{min-entropy} of $X$}.
\end{definition}

The version of this problem where Alice and Bob are not allowed to communicate at all was considered by Bogdanov and Mossel~\cite{BogdanovMossel}. Ahlswede and Csiszar~\cite{AC:98} consider the setting where Alice and Bob must agree on the same string with high probability as the number of correlated samples grows, and pinpoint the ratio between communication required and number of random bits agreed upon. 

A trivial way to distill randomness would be for Alice to toss random
coins and send their outcome to Bob. This would achieve $\gamma = 1$ and
communication complexity of $k$ for $k$ bits of entropy.
Our first proposition {(of which a tighter version was obtained in~\cite{AC:98})} says that with non-trivial correlation, some
savings can always be achieved over this naive protocol.

\begin{proposition}\label{prop:agree-positive}
For every $\rho > 0$, we have
$\owisrcc[\rho](\agreement^k_\gamma) \leq (h(\frac{1-\rho}{2}) + o_k(1))\cdot k$
with $\gamma = 1 - o_k(1)$. In particular for every $\rho > 0$ there exists $\alpha <
1$ such that for every sufficiently large $k$
$\owisrcc[\rho](\agreement^k_{1/2}) \leq \alpha k$.
\end{proposition}

\noindent For completeness, we prove this proposition in \autoref{sec:agreement}. 
\noindent Our next theorem says that these linear savings are the best possible: one
cannot get away with $o(k)$ communication unless $\rho = 1$. For constant $\gamma$ {and the restriction of one-way communication}, this theorem follows from~\cite{AC:98}; since we are here concerned with the case where $\gamma = o(1)$ {even with two-way communication}, we give a proof based on Theorem 1 of \cite{BogdanovMossel} (restated as~\autoref{lem:lb:agreement} here) and a reduction that converts protocols with 
communication to zero-communication protocols with a loss in $\gamma$.

\begin{theorem}\label{th:agreementintro}
$\forall \rho < 1, \exists \eps > 0$ such that
$\isrcc_\rho(\agreement^k_\gamma) \geq \eps k - \frac{3}{2}\log \frac{1}{\gamma} - O(1)$.
\end{theorem}

\noindent \autoref{sec:agreement} contains details of
this proof.

\subsubsection{General relationships between perfect and imperfect sharing}

Our final target in this work is to get some general relationships for communication complexity in the settings of
perfect and imperfectly shared randomness. Our upper bounds for communication complexity are obtained by
considering a natural promise problem, that we call $\gapip$, which is a ``hard problem'' for
communication complexity. We use a variant, $\sparsegapip$, for our lower bounds. We
define both problems below.

\begin{definition}[$\gapip^n_{c,s}$,
$\sparsegapip^n_{q,c,s}$]\label{def:gap:inner-product}
The $\gapip^n_{c,s}$ problem has parameters
$n\in\N$ (dimension), and $c>s\in[0,1]$ (completeness and soundness).
Both $\yes$- and $\no$-instances of this problem have
inputs {$x,y\in\{0,1\}^n$}.
An instance $(x,y)$ is a $\yes$-instance if $\ip{x}{y} \geq cn$,
and a $\no$-instance if $\ip{x}{y} < sn$.
The
$\sparsegapip^n_{q,c,s}$ is a restriction of
$\gapip^n_{c,s}$ where both the \yes- and the
\no-instances are sparse, i.e., $
\normtwo{x}^2 \leq n/q$.
\end{definition}

In \autoref{prop:gapip:hard:psr} we show that
$\gapip^n_{c,s}$ is ``hard'' for $\capspsrcc(k)$
with $c = (2/3) 2^{-k}$ and $s = (1/3)2^{-k}$.
Then in
\autoref{lem:gaussian-upper} we show that this problem is in
$\capsowisrcc[\rho](\poly(1/(c-s)))$. Putting the two results together we get the
following theorem giving a general upper bound on {$\owisrcc[\rho](P)$ in
terms of $\psrcc(P)$} for any promise problem $P$.

\begin{theorem}\label{thm:twoway-psr-isr-upper}
$\forall \rho > 0$, $\exists c < \infty$ such that $\forall k$,
we have
$\capspsrcc(k) \subseteq \capsowisrcc[\rho](c^k)$.
\end{theorem}

\noindent We prove this theorem in \autoref{ssec:upperbound:isr:psr}.

\autoref{thm:twoway-psr-isr-upper} is obviously tight already because of
known gaps between one-way and two-way communication complexity. For
instance, it is well known that the ``indexing'' problem (where Alice gets a vector $x \in \{0,1\}^n$ and Bob an index $i \in [n]$ and they wish to compute $x_i$) has one-way
communication complexity of $\Omega(n)$ with perfectly shared randomness,
while its deterministic two-way communication complexity is at most $\log n
+ 2$.
However one could hope for tighter results capturing promise problems $P$
with low $\owpsrcc(P)$, or to give better upper bounds on $\isrcc(P)$ for
$P$ with low $\psrcc(P)$. Our next theorem rules out any further improvements to
\autoref{thm:twoway-psr-isr-upper} when $n$ is sufficiently large (compared to $k$).
We do so by focusing on the problem $\sparsegapip$.
In \autoref{prop:sparsegapinnerpsr} we show that
$\owpsrcc(\sparsegapip_{q,c,s}^n) = O(\poly(\frac{1}{q(c-s)})\log q)$ for every
$q$, $n$ and $c > s$. In particular if say $c = 1/(2q)$ and $s = 1/(4q)$ the
one-way communication complexity with perfectly shared randomness reduces to
$O(\log q)$, in contrast to the $\poly(q)$ upper bound on the
one-way communication complexity with imperfectly shared randomness
from \autoref{lem:gaussian-upper}.

Our main technical theorem shows that this gap is necessary
for every $\rho < 1$. Specifically in \autoref{thm:main-neg} we show
that
{$\isrcc_\rho(\sparsegapip_{q,c=.9/q,s = .6/q}^n) = \Omega(\sqrt{q})$}.
Putting the two together we get a strong converse to
\autoref{thm:twoway-psr-isr-upper},
stated below.

\begin{theorem}\label{thm:psr-isr-lower}
For every $k$, there exists a promise problem
$P = (P_n)_{n \in \N}$ such that $\owpsrcc(P) \leq k$, but for every
$\rho < 1$ it is the case that $\isrcc_\rho(P) = 2^{\Omega_\rho(k)}$.
\end{theorem}

\paragraph{Remarks on the proofs.}
\autoref{thm:twoway-psr-isr-upper} and \autoref{thm:psr-isr-lower}
are the technical highlights of this paper and
we describe some of the ideas behind them here.

\autoref{thm:twoway-psr-isr-upper} gives an upper bound for $\owisrcc[\rho]$ for
problems with low {$\psrcc$}. As such this ought to be somewhat surprising in
that for known problems with low probabilistic communication complexity
(notably, equality testing), the known solutions are very sensitive to perturbations of the randomness.
But the formulation in terms of {\gapip} suggests that any such problem
reduces to an approximate inner product calculation; and the theory of metric
embeddings, and examples such as locality sensitive hashing, suggest that one
can reduce the dimensionality of the problems here significantly and this may
lead to some reduced complexity protocols that are also robust to the noise of
the $\rho$-correlated vectors. This leads us to the following idea: To
estimate $\ip{x}{y}$, where $x,y \in {\{0,1\}^n}$, Alice can compute
$a = \ip{g_1}{x}$ where $g_1$ is a random $n$-dimensional spherical Gaussian
and send $a$ (or the most significant bits of $a$) to Bob. Bob can compute $b
= \ip{g_2}{y}$ and $a\cdot b$ is {an unbiased} estimator (up to normalization) of
$\ip{x}{y}$ if $g_1 = g_2$. This protocol can be easily shown to be robust in
that if $g_2$ is only $\rho$-correlated with $g_1$, $a\cdot b$ is still a good
estimator, with higher variance. And it is easy to convert a collection of $\rho$-correlated
bits to $\rho$-correlated Gaussians, so it is possible for Alice and Bob to generate
the $g_1$ and $g_2$ as desired from their imperfectly shared randomness.
A careful analysis (of a variant of this protocol) shows that to estimate
$\ip{x}{y}$ to within an additive error $\eps \normtwo{x} \normtwo{y}$, it
suffices for Alice to send about $1/\eps^2$ bits to Bob, and
this leads to a proof of \autoref{thm:twoway-psr-isr-upper}. 
Next we turn to the proof of \autoref{thm:psr-isr-lower}, which shows a roughly matching
lower bound to \autoref{thm:twoway-psr-isr-upper} above. The insight to this proof comes from examining the ``Gaussian protocol''
above carefully and contrasting it with the protocol used in the perfect
randomness setting. In the latter case Alice uses the randomness to pick one
(or few) coordinates of $x$ and sends some function of these bits to Bob
achieving a communication complexity of roughly $\log(1/\eps)$, using the
fact that only $O(\eps n)$ bits of $x$ are non-zero.
In the Gaussian protocol Alice sends a very ``non-junta''-like function
of $x$ to Bob; this seems robust to the perturbations of the randomness,
but leads to $1/\eps^2$ bits of communication.
This difference in behavior suggests that perhaps functions where variables
have low ``influence'' cannot be good strategies in the setting of perfect
randomness, and indeed we manage to prove such a statement in
\autoref{thm:singlefg}. The proof of this theorem uses a variant of the invariance
principle that we prove (see \autoref{thm:invariance:principle:bifunction}), which shows that if a
communication protocol with
low-influences works in a ``product-distributional'' setting, it will also work with inputs being Gaussian and with the same moments.
This turns out to be a very useful reduction. The reason that
{\sparsegapip} has nice $\owpsrcc$ protocols is the asymmetry between the inputs
of Alice and the inputs of Bob --- inputs of Alice are sparse! But with the
Gaussian variables there is no notion of sparsity and indeed Alice and Bob
have symmetric inputs and so one can now reduce the {``disjointness''} problem
from communication complexity (where now
{Alice and Bob hold sets $A,B\subseteq[1/\eps]$, and would like to distinguish $\abs{A\cap B}=0$ from $\abs{A\cap B}=1$)}
to the Gaussian inner product problem. Using the
well-known lower bound on {disjointness}, we conclude that {$\Omega(1/\eps)$ bits} of communication are necessary and this proves \autoref{thm:singlefg}.

Of course, all this rules out only one part of the solution space for the
communication complexity problem, one where Alice and Bob use functions of
low-influence. To turn this into a general lower bound we note that if Alice
and Bob use functions with some very influential variables, then they should
agree on which variable to use (given their randomness $r$ and $r^\prime$). Such
agreement on the other hand cannot happen with too high a probability by
our lower bound on \agreement (from \autoref{th:agreementintro}). Putting all these
ingredients together gives us a proof of \autoref{thm:psr-isr-lower} (see
\autoref{ssec:lowerbound:isr}) for more details).

\paragraph{Organization of the rest of the paper} The rest of the paper contains details and proofs of the theorems mentioned in this section. In the next section (\autoref{sec:compress}), we prove our isr upper bound for the ``Uncertain Compression'' problem, {namely} \autoref{thm:compress}. We then turn, in \autoref{sec:agreement}, to the matching upper and lower bounds for ''Agreement Distillation'' {as described in} \autoref{prop:agree-positive} and \autoref{th:agreementintro}. \autoref{sec:general} contains the details of our main results relating communication with {perfectly} and {imperfectly} shared randomness, \autoref{thm:twoway-psr-isr-upper} and \autoref{thm:psr-isr-lower}: we first describe an alternate characterization of communication strategies in \autoref{ssec:strategy:preliminaries}, which allows us to treat them as vectors in (carefully defined) convex sets. This enables us to use ideas and machinery from Gaussian analysis: in particular, our lower bound on isr presented in \autoref{sec:proof:theo:main-neg} relies on a new invariance theorem, \autoref{thm:invariance:principle:bifunction}, that we prove in \autoref{sec:invariance}.

\section{Compression}\label{sec:compress}

In this section, we prove \autoref{thm:compress}, restated below:
\thmcompress*

\begin{proof}[Proof of \autoref{thm:compress}]
Let $\mu = (1-\rho)/2$ and $\eps'>0$ be such that
$1/(1-h(\mu + \eps')) = (1+\eps)/(1-h(\mu))$.
Let $c = O(\frac{1}{{\eps'}^2}\ln(1/\delta))$.

We interpret the random strings $r$ and $r^\prime$ as two ``dictionaries'', i.e.,
as describing words $\{w_{i,j} \in \{-1,+1\}^j\}_{{i \in [n],j\in\N}}$
and $\{w'_{i,j} \in \{-1,+1\}^j\}_{{i \in [n],j\in\N}}$, with the property that for every
$i,j$ and coordinate $k \in [j]$, the $k$th coordinates of $w_{i,j}$
and $w'_{i,j}$ are $\rho$-correlated.

\noindent On input $P,m$ Alice sends $X = w_{m,j}$ to Bob
where $j = \max\{c, \frac{1+\eps}{1 - h(\mu)}\left(\log (1/P(m)) + 2\Delta + \log(1/\delta)\right) \}$. On input $Q$ and on receiving $X$ from Alice, Bob computes
$j = |X|$ and the set
\[ S_X = \setOfSuchThat{\tilde{m} }{ \mathrm{dist}(w'_{\tilde{m},j},X) \leq (\mu + \eps')j } \ , \]
 where $\mathrm{dist}$ denotes the Hamming distance between strings. Bob then outputs $\mathrm{argmax}_{\tilde{m} \in S_X} \{Q(\tilde{m})\}$ (so it outputs
the most likely message after some error-correction).

It is clear from construction that the expected length of the
communication when $m \sim P$ is at most
\begin{multline*}
\E_{m \sim P} \left[\frac{1+\eps}{1 - h(\mu)}\left(\log (1/P(m)) + 2\Delta +
c\right)\right]
= \\ \frac{1+\eps}{1 - h(\mu)}\left( \E_{m \sim P} [\log (1/P(m))] + 2\Delta + c\right)
= \frac{1+\eps}{1 - h(\mu)}\left( H(P) + 2\Delta + c\right).
\end{multline*}

We finally turn to correctness, i.e., to show that Bob's output $\tilde{m} =
m$ with probability at least $1-\delta$. First note that the probability
that $m \in S_X$ is at least $(1-\delta/2)$ (by a simple application of
Chernoff bounds and the fact that $j$ is sufficiently large compared to
$\eps'$ and $\delta$).
Now let $T_m = \setOfSuchThat{m'\ne m }{ P(m') \geq P(m)/4^\Delta}$. Note that $|T_m|\leq
4^\Delta/P(m)$. For any fixed $m' \in T_m$, we have that the probability
(over the choice of $w'_{m',j}$) that $m' \in S_X$ is at most
$2^{-(1-h(\mu+\eps'))j}$. \newer{Indeed, since $m' \neq m$, the two sets of indices corresponding to $w'_{m',j}$ and $w_{m,j}$ (in the ``random bit dictionaries'' $r,r'$) are disjoint, and so $w'_{m',j}$, $w_{m,j}$ are independent.
Now, for any two independent u.a.r. $u,v$, the probability that $\dist{u}{v} \leq \alpha j$ is given by
\[
  \sum_{i \leq \alpha j} \binom{j}{i} \frac{1}{2^{i}}\frac{1}{2^{j-i}}
  = \frac{1}{2^j}\sum_{i \leq \alpha j} \binom{j}{i}
  \leq \frac{1}{2^j}\sum_{i \leq \alpha j} \binom{j}{i}
  \leq 2^{-j}2^{h(\alpha)j}.
\]
}
Taking the union bound over $m' \in T_m$ and
plugging in our choice of $j$, we have
that with probability at least $1 - \delta/2$, $T_m \cap S_X = \emptyset$.
With probability at least $1 - \delta$ both events above happen and when
they do, as {$m\in S_X$ satisfies $Q(m) \geq \frac{P(m)}{2^\Delta}$ and any other element $m^\prime\in S_X$ is such that $Q(m^\prime) \leq 2^\Delta P(m^\prime) < 2^\Delta P(m)/4^\Delta$}, we have $\tilde{m} = m$.
\end{proof}

\section{Agreement Distillation}\label{sec:agreement}

In this section we give proofs of \autoref{prop:agree-positive} and
\autoref{th:agreementintro}
which respectively give upper and lower bounds on the one-way
communication complexity of randomness distillation.

We start with the upper bound, which relies on the existence of
linear error-correcting codes capable of correcting $\mu \triangleq
\frac{1-\rho}2$ fraction errors. The fact that such codes have rate
approaching
$1 - h(\mu)$ yields the result that agreement distillation requires
$(1 + o_k(1))\cdot h(\mu) \cdot k$ communication for $\gamma \to 1$.
Details below.

\begin{proof}[Proof of \autoref{prop:agree-positive}] Let $\eps>0$ be any positive constant and let $\bern{k}{\mu}$ be the distribution on $\{0,1\}^{k}$ where each bit is independent and is $1$ with probability $\mu$. {Our proof protocol will rely on the existence of a certain matrix $H \in \{0,1\}^{\ell \times k}$ over $\mathbb{F}_2$ satisfying the following property:
\begin{equation}\label{eq:condition:matrix:agreement}
  \Pr_{e \sim \bern{k}{\mu}} \left[\exists e'\ne e \mbox{ s.t. } \wt(e') \leq
(\mu+\eps)k \mbox{ and } H\cdot e' = H \cdot e\right] \leq \delta/2
\end{equation}
(we will establish the existence of such a matrix later, for $\ell = h(\mu + \eps)k$).

With this matrix $H$ in hand, Alice and Bob act as follows.} Given $\rho$ correlated strings $r,r' \in \{0,1\}^k$, Alice's output is
$w_A = r$. She communicates $y=H\cdot r$ to Bob. Bob's output is $w_B =
\tilde{r}$ such that (i) $H\cdot \tilde{r} = y$ and
(ii) $\dist{\tilde{r}}{r'} \leq (\mu+\eps)k$, provided
$\tilde{r}$ with these properties exists and is unique. Else he outputs $r'$.

It follows that unless $\dist{{r}}{r'} > (\mu+\eps)k$ or if $\exists e' \neq e \triangleq r-r'$ such that $\wt(e') \leq
(\mu+\eps)k$  and $H\cdot e' = H \cdot e$, we have $\tilde{r} = r$. {(Indeed, this is a consequence of the above property of $H$, writing $e' \triangleq \tilde{r}-r'$ and observing that $\wt(e')=\dist{\tilde{r}}{r'}$, and that $He' = He$ if and only if $H\cdot \tilde{r}=H\cdot r$.)}
The probability of either event above is small (by Chernoff bound for
the first, and by the condition on $H$ for the second). 

It remains to prove the existence of a matrix $H$ satisfying~\eqref{eq:condition:matrix:agreement}, for small $\ell\in\N$. 
We will actually show that a random matrix satisfies this condition for $\ell = h(\mu +
\eps)k$ with probability tending to $1$ as $k$ goes to $\infty$: indeed, fixing any non-zero vector $x\in\{0,1\}^k$, a random matrix $H$ obtained by setting independently each coefficient to be $1$ with probability $1/2$ satisfies $\probaOf{Hx = 0^\ell} = 2^{-\ell}$. The claim follows from a union bound over the (at most) $\sum_{i\leq (\mu+\eps)k} \binom{k}{i} \leq 2^{h(\mu+\eps)k - \omega_k(1)}$ vectors $e^\prime$.
\end{proof}

We now turn towards a proof of \autoref{th:agreementintro}. We first
consider the setting of {\em zero} communication, i.e., when Alice and Bob
are not allowed to communicate at all. Here we use the following lemma
due to \cite{BogdanovMossel} which shows that the
agreement probability $\gamma$ is exponentially small in $k$.

\begin{lemma}[{\cite[Theorem 1]{BogdanovMossel}}]\label{lem:lb:agreement}
$\forall \rho < 1, \exists\eps > 0$ such that for every zero-communication protocol for \linebreak$\agreement^k_\gamma$, we have
$\gamma \leq 2^{-\eps k}$. 
(Furthermore, one can take $\eps = 1-\bigO{\rho}$). \end{lemma}

\noindent We now derive \autoref{th:agreementintro} as a corollary of \autoref{lem:lb:agreement}. 
\begin{proof}[Proof of \autoref{th:agreementintro}]
       Suppose $\Pi$ is a $c$-bit communication protocol for $\agreement^k_\gamma$. We claim we can convert $\Pi$ to a zero-bit communication protocol $\Pi'$ for $\agreement^k_{\gamma'}$, where the agreement $\gamma'$ is $\poly(\gamma,2^{-c})$.  To do so, let $(r,s)$ denote the randomness inputs to Alice and Bob. Let $\Pi_A(r,s)$ and $\Pi_B(r,s)$ denote Alice's and Bob's outputs, respectively. The zero-communication protocol $\Pi'$ is obtained as follows.
\begin{itemize}
  \item $\Pi'_A(r)$: Alice samples $s'$ conditioned on $r$ and outputs $\Pi_A(r,s')$;
  \item $\Pi'_B(s)$: Bob similarly samples $r'$ conditioned on $s$ and outputs $\Pi_B(r',s)$.
\end{itemize}
Since the input distributions in the invocation of $\Pi_A$ and $\Pi_B$ {(by $\Pi'_A$ and $\Pi'_B$, respectively)} are exactly the same as in the $c$-communication protocol, the entropy of the output is unchanged. So it suffices to argue that agreement happens with probability $\poly(\gamma, 2^{-c})$.

Let $P(r,s)$ denote the probability of the input being $(r,s)$. Fix the private randomness $\theta\eqdef (\theta_A,\theta_B)$ of Alice and Bob, and let $\gamma_{\theta}$ be the agreement probability for this fixed choice of the randomness of Alice and Bob. From now on we can consider $\Pi$ to be a deterministic function of $r$ and $s$. Let $t(r,s)$ denote the transcript under $\Pi$ on input $(r,s)$, and $Q_r(t)$ (resp. $Q_s(t)$) denote the probability that the transcript of $\Pi$ equals $t\in \{0,1\}^c$ conditioned on Alice's input being $r$ (resp. Bob's input being $s$). Let $G$ be the subset of values of $(r,s)$ on which $\Pi$ has agreement, so that $\sum_{(r,s) \in G} P(r,s) = \gamma_{\theta}$.

Note that the agreement $\gamma'_{\theta}$ under $\Pi'$ (on randomness $\theta$) can be lower bounded by $\gamma_{\theta}'\geq \sum_{(r,s) \in G} P(r,s) \cdot Q_r(t(r,s)) \cdot Q_s(t(r,s))$, which is the probability with which Alice and Bob generate the same transcript in $\Pi'$ as they would under $\Pi$. Now, had all the $Q_r(t)$'s been equal to $2^{-c}$, we would be immediately done with a lower bound of $\gamma_{\theta}/4^{-c}$. We shall obtain a slightly worse bound in order to handle non-uniform distributions.

Say the transcript $t$ is \emph{unlikely for $r$} if $Q_r(t) < \frac{\gamma_{\theta}}{4}2^{-c}$. We consider the set of ``unlikely randomness'' $B \eqdef \setOfSuchThat{(r,s) }{ t(r,s) \text{ is unlikely for } r \text{ or is unlikely for } s}$.
We first note that $\sum_{(r,s) \in B} P(r,s) < \frac{\gamma_{\theta}}{2}$. To see this, observe that
\begin{align*}
\sum_{(r,s) : t(r,s) \text{ unlikely for } r} P(r,s)
&= \sum_r \sum_{t : t \text{ unlikely for } r} \sum_{s : t(r,s) = t} P(r,s) \\
&= \sum_r P(r) \sum_{t : t  \text{ unlikely for } r} Q_r(t) \\
&<  \sum_r P(r) \sum_{t : t  \text{ unlikely for } r} \frac{\gamma_{\theta}}{4}2^{-c} \\
&< \frac{\gamma_{\theta}}{4}
\end{align*}
Similarly, we have $\sum_{(r,s) : t(r,s) \text{ unlikely for } s} P(r,s)< \frac{\gamma_{\theta}}{4}$, yielding the claimed bound.

\noindent From there, we can write
\begin{align*}
  \gamma'_{\theta} &\geq \sum_{(r,s) \in G} P(r,s) \cdot Q_r(t(r,s)) \cdot Q_s(t(r,s))\\
    &\geq \sum_{(r,s) \in G \setminus B} P(r,s) \cdot Q_r(t(r,s)) \cdot Q_s(t(r,s)) \\
    &\geq \sum_{(r,s) \in G \setminus B} P(r,s) \cdot \left( \frac{\gamma_{\theta}}{4}2^{-c} \right)^2 \\
    &= \frac{\gamma^2_{\theta}}{16} 4^{-c} \left( \sum_{(r,s) \in G} P(r,s) - \sum_{(r,s) \in B} P(r,s) \right)   \\
    &\geq \frac{\gamma^3_{\theta}}{32} 4^{-c}
\end{align*}

\noindent Finally, we take expectations over the private randomness $\theta$: this leads to
\[
  \shortexpect_\theta[\gamma'_\theta] \geq \shortexpect_\theta\left[\frac{\gamma_\theta^3}{32}\cdot 4^{-c}\right] \geq \frac{\shortexpect_\theta[\gamma_\theta]^3}{32}\cdot 4^{-c} = \frac{\gamma^3}{32}\cdot 4^{-c},
\]
as claimed. Applying \autoref{lem:lb:agreement}, we get that for some constant $\eps^\prime = 1-O(\rho) \in (0,1)$, $2^{-2c}\frac{\gamma^3}{32} \leq 2^{-\eps^\prime k}$ and thus $c \geq \frac{\eps^\prime}{2}k - \frac{3}{2}\log\frac{1}{\gamma} - O(1)$ as desired. Taking $\eps\eqdef \frac{\eps^\prime}{2}$ concludes the proof.
\end{proof}

\section{General connection between {perfectly} and {imperfectly} shared randomness}
\label{sec:general}

In this section we present proofs of
\mnew{
\autoref{thm:twoway-psr-isr-upper} and
\autoref{thm:psr-isr-lower}}.
Key to both our upper bound on $\owisrcc(P)$ in terms of
$\psrcc(P)$, and our lower bound on $\isrcc(\sparsegapip)$,
is a representation of communication strategies as vectors, \mnew{where} the
success probability of an interaction \mnew{is proportional} to the inner product of these vectors.
We describe this representation in \autoref{ssec:strategy:preliminaries} below. We then
use this representation to show that \gapip is hard for
$\capspsrcc(k)$ in \autoref{ssec:upperbound:isr:psr}. We also give a one-way isr protocol
for \gapip in the same section thus giving a proof of
\autoref{thm:twoway-psr-isr-upper}.
Finally in \autoref{ssec:lowerbound:isr} we give a one-way psr protocol for
\sparsegapip, and then state our main technical result ---  an exponentially higher
lower bound for it in the two-way isr setting (with the proof deferred to \autoref{sec:proof:theo:main-neg} modulo an invariance principle which is established in \autoref{sec:invariance}). The lower bound uses the fact
that the space of strategies in the vector representation forms a bounded convex
set.

\subsection{Communication Strategies: Inner Products and Convexity}
\label{ssec:strategy:preliminaries}

We start by formalizing deterministic and probabilistic (private-coin) two-way communication strategies
for Alice and Bob. By ``strategy'' we mean what Alice would do given her input and
randomness, as a function of different messages that Bob may send her, and
vice versa.
We restrict our attention to canonical protocols in
which Alice and Bob strictly alternate and communicate {one bit per round}; and
the eventual outcome is a Boolean one, determined after
$k$ rounds of communication. (So the only problems that can be solved this way are ``promise problems''.)
Without loss of generality we also assume that the
last bit communicated is the output of the communication protocol.

The natural way to define deterministic strategies would be in terms of a triple
$(f_A,f_B,v)$ where
$f_A = (f^{2i}_A\colon\{0,1\}^{2i} \to \{0,1\})_{0 \leq i < k/2}$ is a sequence of functions
and so is
$f_B = (f^{2i+1}_B\colon\{0,1\}^{2i+1} \to \{0,1\})_{0 \leq i < k/2}$
and $v\colon\{0,1\}^k \to \{0,1\}$.
The function $f^{2i}_A(h)$ determines Alice's message bit after $2i$ rounds of
communication, with $h \in \{0,1\}^{2i}$ being the transcript of the interaction thus
far.
Similarly the functions $f^{2i+1}_B(h)$ determine Bob's message bit after $2i+1$ rounds of communication.
Finally, $v$ denotes the verdict function. Since we assumed that the last
bit transmitted is the output, we have $v(\ell_1,\ldots,\ell_k) = \ell_k$.
Thus the output of an interaction is given by
$v(\ell)$ where $\ell = (\ell_1,\ldots,\ell_k)$ is
given by $\ell_{2i+1} = f_A^{2i}(\ell_1,\ldots,\ell_{2i})$ and
$\ell_{2i+2} = f_B^{2i+1}(\ell_1,\ldots,\ell_{2i+1})$ for $0\leq i \leq k/2$.
The interpretation is that Alice
can determine the function $f_A$ from her input and Bob can determine $f_B$ from his input,
and this allows both to determine the output after $k$ rounds of interaction.

We will be moving on to the vector representation of strategies shortly, but
first we describe probabilistic interactions, where Alice and Bob have private
randomness.\footnote{{Since we do not concern ourselves with computational complexity of the protocol, we can assume without loss of generality that the players use fresh randomness at every stage. Indeed, both parties, if willing to rely on some of their respective ``previous randomness,'' can instead sample new random bits conditioned on the past transcript.}} Such an interaction is also described by a triple $(f_A,f_B,v)$ except
that now $f_A = (f^{2i}_A\colon\{0,1\}^{2i} \to [0,1])_{0 \leq i < k/2}$ and
$f_B = (f^{2i+1}_B\colon\{0,1\}^{2i+1} \to [0,1])_{0 \leq i < k/2}$. The outcome is now
the random variable $v(\ell)$ where $\ell = (\ell_1,\ldots,\ell_k)$ is the random variable
determined inductively by letting
$\ell_{2i+1} = 1$ with probability $f_A^{2i}(\ell_1,\ldots,\ell_{2i})$ and
$\ell_{2i+2} = 1$ with probability $f_B^{2i+1}(\ell_1,\ldots,\ell_{2i+1})$ for $0\leq i \leq k/2$.

Our vector representation of deterministic interactions is obtained by
considering the set of
``plausible final transcripts'' that a player might see given their own
strategy.
Recall that the transcript of an interaction is a $k$-bit string and there are $2^k$ possible transcripts.
In the new representation, we represent Alice's strategy (i.e., the functions $f_A$) by a
vector $\strat{A} \in \{0,1\}^{2^k}$
where $\strat{A}(\ell) = 1$ if and only if $\ell \in \{0,1\}^k$ is a transcript {\em consistent} with
Alice's strategy.
(We give a more formal description shortly.)
For probabilistic communication strategies
(corresponding to Alice and Bob working with private randomness),
we represent them by vectors $\strat{A}$ and $\strat{B}$ in
$[0,1]^{2^k}$. We formalize the set of such strategies, and verdicts, below.

In what follows we describe subsets of {$[0,1]^{2^k}$} that are supposed to describe the strategy space for Alice and Bob. Roughly,
we wish to allow $\strat{A} = (\strat{A}(i_1,\ldots,i_k))_{i_1,\ldots,i_k \in \{0,1\}}$ to be
an ``Alice strategy'' (i.e., a member of $K_A$) if for every
$i_1,\ldots,i_k$ there exists a Bob strategy such that Alice reaches the transcript $i_1,\ldots,i_k$ with probability $\strat{A}(i_1,\ldots,i_k)$.
To describe this set explicitly we introduce auxiliary variables
$\ptranscript{A}(i_1,\ldots,i_j)$ for every $0 \leq j \leq k$ and $i_1,\ldots,i_j \in
\{0,1\}$ where $\ptranscript{A}(i_j,\ldots,i_j)$ denotes the probability (again
maximized over Bob strategies) of reaching the partial transcript
$i_1,\ldots,i_j$. In what follows we first show that the auxiliary variables are linear forms in $\strat{A}$ and then show the conditions that the auxiliary
variables satisfy. (We warn the reader that the first step --- showing that the
$\ptranscript{A}(\cdots)$'s are linear forms in $\strat{A}$ -- relies on the constraints imposed
later and so some of the definition may be slightly non-intuitive.)
Together the two steps allows us to show that the space of strategies is a
{(closed)} convex set.

\def\withpt{1} \ifnum\withpt=1
\newcommand{\pt}{\operatorname*{PT}}
\begin{definition}\label{def:xa:xb}
We define the \emph{partial transcript operators} $\pt_A,\pt_B\colon [0,1]^{2^k} \to ([0,1]^{\{0,1\}^{j}})_{0\leq j \leq k}$, which map a vector $\strat{} \in [0,1]^{2^k}$ to respectively Alice and Bob strategies $\pt_A(\strat{}),\pt_B(\strat{})$. 
For vector $\strat{} \in [0,1]^{2^k}$, $0\leq j\leq k$, and $i_1,\ldots,i_j \in \{0,1\}$
let $(\pt_A(\strat{}))_j(i_1,\ldots,i_j) = 
\ptranscript{A}(i_1,\ldots,i_j) \in [0,1]$ and 
$(\pt_B(\strat{}))_j(i_1,\ldots,i_j) = 
\ptranscript{B}(i_1,\ldots,i_j) \in [0,1]$ be defined as
follows:
\begin{align*}
 \ptranscript{A}(i_1,\ldots,i_j) & = \begin{cases}
				\strat{}(i_1,\ldots,i_k) &\text{ if } j=k\\
				\ptranscript{A}(i_1,\ldots,i_j,0) + \ptranscript{A}(i_1,\ldots,i_j,1) &\text{ if } j \text{ is even.}\\
				\frac12\left(\ptranscript{A}(i_1,\ldots,i_j,0) + \ptranscript{A}(i_1,\ldots,i_j,1)\right) &\text{ if } j \text{ is odd.}\\
			\end{cases}
\end{align*}
\begin{align*}
 \ptranscript{B}(i_1,\ldots,i_j) & = \begin{cases}
				\strat{}(i_1,\ldots,i_k) &\text{ if } j=k\\
				\frac12\left(\ptranscript{B}(i_1,\ldots,i_j,0) + \ptranscript{B}(i_1,\ldots,i_j,1)\right) &\text{ if } j \text{ is even.}\\
				\ptranscript{B}(i_1,\ldots,i_j,0) + \ptranscript{B}(i_1,\ldots,i_j,1) &\text{ if } j \text{ is odd.}\\
			\end{cases}
\end{align*}

\noindent Define
\[\tilde{K}_A = \tilde{K}^{(k)}_A = \setOfSuchThat{ \strat{} \in [0,1]^{2^k} }{ \ptranscript{A}() = 1 \text{ and }\forall \text{ odd
}j, \forall i_1,\ldots,i_j \in \{0,1\}, ~~
				\ptranscript{A}(i_1,\ldots,i_j,0) =
				\ptranscript{A}(i_1,\ldots,i_j,1) },
				\]
and
\[\tilde{K}_B = \tilde{K}^{(k)}_B = \setOfSuchThat{\strat{} \in [0,1]^{2^k} }{ \ptranscript{B}() = 1 \text{ and }\forall \text{
even }j, \forall i_1,\ldots,i_j \in \{0,1\}, ~~
				\ptranscript{B}(i_1,\ldots,i_j,0) =
				\ptranscript{B}(i_1,\ldots,i_j,1) }.
				\]
Let $K_A = \setOfSuchThat{\strat{} \ast v }{ \strat{} \in \tilde{K}_A }$, where $v \in \{0,1\}^{2^k}$ is
given by $v(i_1,\ldots,i_k) = i_k$ (and $a \ast b$ denotes coordinate-wise
multiplication of vectors $a$ and $b$).
Let
$\tilde{S}_A = \tilde{K}_A \cap \{0,1\}^{2^k}$,
$\tilde{S}_B = \tilde{K}_B \cap \{0,1\}^{2^k}$,
$S_A = K_A \cap \{0,1\}^{2^k}$,
and $S_B = K_B \cap \{0,1\}^{2^k}$.
\end{definition}

\begin{figure}[h]\centering
\scalebox{0.8}{
    \begin{tikzpicture}[->,>=stealth',
      level 1/.style={sibling distance=85mm},
      level 2/.style={sibling distance=40mm},
      level 3/.style={sibling distance=0mm},
      level distance=20mm,
      ]
    \node {$\underbrace{\ptranscript{A}(i_1,\dots,i_{2j})}_{\color{black!65}0+1=1}$}
        child{ node {$\underbrace{\ptranscript{A}(i_1,\dots,i_{2j},i_{2j+1}=0)}_{\color{black!65}\frac{1}{2}(0+0)=0}$}
                child{ node {$\underbrace{\ptranscript{A}(i_1,\dots,i_{2j},i_{2j+1},0)}_{\color{black!65}0}$} edge from parent node[above left] {$0$}
                }
                child{ node {$\underbrace{\ptranscript{A}(i_1,\dots,i_{2j},i_{2j+1},1)}_{\color{black!65}0}$} edge from parent node[above right] {$1$}
                }
           edge from parent node[above left] {$0$}
        }
        child{ node {$\underbrace{\ptranscript{A}(i_1,\dots,i_{2j},i_{2j+1}=1)}_{\color{black!65}\frac{1}{2}(1+1)=1}$}
                child{ node {$\underbrace{\ptranscript{A}(i_1,\dots,i_{2j},i_{2j+1},0)}_{\color{black!65}1}$} edge from parent node[above left] {$0$}
                }
                child{ node {$\underbrace{\ptranscript{A}(i_1,\dots,i_{2j},i_{2j+1},1)}_{\color{black!65}1}$} edge from parent node[above right] {$1$}
                }
          edge from parent node[above right] {$1$}
		    }
    ;
    \end{tikzpicture}
}
\caption{\label{fig:constraints:xa} Illustration of the constraints on $\ptranscript{A}$ (\autoref{def:xa:xb}).}
\end{figure}

In what follows we first focus on deterministic communication strategies
and show that $\tilde{S}_A$, $\tilde{S}_B$ correspond to
the space of deterministic communication strategies for Alice and Bob,
while $S_A$ and $S_B$ correspond to outputs computed by such strategies.
This step is not strictly needed for this paper since our main focus is
on probabilistic strategies and the convex sets $K_A$ and $K_B$, but
we include it for completeness.

\begin{proposition}
\label{prop:sa:sb}.
$\tilde{S}_A$ and $\tilde{S}_B$ correspond to the set of deterministic communication
strategies with $k$ bits. For every strategy $f_A$ of Alice there exist
vectors
$\tildestrat{A} \in \tilde{S}_A$ and $\strat{A}\in S_A$ and for every strategy $f_B$ of Bob there exist  vectors $\tildestrat{B} \in
\tilde{S}_B$ and $\strat{B} \in S_B$
such that if $\ell \in \{0,1\}^k$ is the transcript of the interaction between Alice and Bob
under strategies $f_A$ and $f_B$, then $\ell$ is the unique sequence satisfying 
$\tildestrat{A}(\ell) = \tildestrat{B}(\ell) = 1$ and $\ip{\strat{A}}{\strat{B}} = 1$ if the interaction
accepts and $0$ otherwise.

Conversely every vector $\strat{A} \in S_A$ corresponds to a strategy $f_A$ for Alice
(and similarly for Bob) such that Alice and Bob accept the interaction iff $\ip{\strat{A}}{\strat{B}} = 1$.
\end{proposition}
\begin{proof}
Given $f_A$ to construct $\tildestrat{A}$, we let $\tildestrat{A}(\ell) = 1$ if there exists $f_{B,\ell}$ such that the final
transcript of the interaction given by $f_A$ and $f_{B,\ell}$ is $\ell$.
Furthermore let $\tildeptranscript{A}(i_1,\ldots,i_j) = 1$ if there exists a Bob
strategy $f_{B,i_1,\ldots,i_j}$ such that $i_1,\ldots,i_j$ is the partial
transcript of the interaction between Alice and Bob; {otherwise let $\tildeptranscript{A}(i_1,\ldots,i_j) = 0$}.
It is now straightforward to verify that the $\tildeptranscript{A}(i_1,\ldots,i_j)$ satisfy the conditions of the definition of $\tildestrat{A}$ and the conditions required for
membership in $\tilde{K}_A$. In particular we have the following three conditions: (1)
$\tildeptranscript{A}() = 1$ since the empty transcript is a legal partial transcript.  (2) If $j$ is an
even index (and so Alice speaks in round $j+1$) and $\ptranscript{A}(i_1,\ldots,i_j) = 0$ (so the partial
transcript $i_1,\ldots,i_j$ is not reachable given Alice's strategy), then we must have
$\ptranscript{A}(i_1,\ldots,i_j,0) = \ptranscript{A}(i_1,\ldots,i_j,1) = 0$ (no extension is reachable either). If
$\ptranscript{A}(i_1,\ldots,i_j) = 1$ then exactly one of the extensions must be reachable (based on Alice's
message at this stage) and so again we have
$\ptranscript{A}(i_1,\ldots,i_j) = \ptranscript{A}(i_1,\ldots,i_j,0) + \ptranscript{A}(i_1,\ldots,i_j,1)$. (3) If $j$ is odd and it is
Bob's turn to speak, then again if $\ptranscript{A}(i_1.\ldots,i_j) = 0$ we have
$\ptranscript{A}(i_1,\ldots,i_j,0) = \ptranscript{A}(i_1,\ldots,i_j,1) = 0$. On the other hand if
$\ptranscript{A}(i_1,\ldots,i_j) = 1$ then for each extension there exists a strategy of Bob that permits this
extension and so we have
$\ptranscript{A}(i_1,\ldots,i_j,0) = \ptranscript{A}(i_1,\ldots,i_j,1) = 1$ satisfying the condition for odd $j$.
The above three conditions verify membership in $\tilde{K}_A$ and since $\tildestrat{A}$ is a 0/1 vector, we also have
$\tildestrat{A} \in \tilde{S}_A$. The vector $\strat{A} = \tildestrat{A} \ast v$ gives the
corresponding vector in $S_A$.

Given a pair of strategies $f_A$ and $f_B$, let 
$\tildestrat{A}$ and $\tildestrat{B}$ be the corresponding vectors representing the strategies, and let 
$\tildeptranscript{A}(i_1,\ldots,i_j) = (\pt_A(\tildestrat{A})_j)(i_1,\ldots,i_j)$ 
and 
$\tildeptranscript{B}(i_1,\ldots,i_j) = (\pt_A(\tildestrat{B})_j)(i_1,\ldots,i_j)$ denote the partial transcripts.
We now prove the existence and uniqueness of a leaf $\ell$ such that 
$\tildestrat{A}(\ell) = \tildestrat{B}(\ell) = 1$.
We do so by showing, by induction on $j$, that for every $j \in \{0,\ldots,k\}$
there exists a unique sequence $(i_1,\ldots,i_j)$ such
that $\tildeptranscript{A}(i_1,\ldots,i_j) = \tildeptranscript{B}(i_1,\ldots,i_j) =
1$. Using the fact that 
$\tildeptranscript{A}(i_1,\ldots,i_k) = 
\tildestrat{A}(i_1,\ldots,i_k)$,
we get the desired existence and uniqueness (for $\ell = (i_1,\ldots,i_k)$).
We refer to a sequence $(i_1,\ldots,i_j)$ as a valid (partial) transcript if
$\tildeptranscript{A}(i_1,\ldots,i_j) = \tildeptranscript{B}(i_1,\ldots,i_j) = 1$
and as an invalid transcript otherwise.
The base case of the induction is true since there is only one transcript of
length $0$ and 
we have $\tildeptranscript{A}() = \tildeptranscript{B}() = 1$ so the empty
transcript is the unique valid transcript.
Now assume the statement is true for transcripts of length $j-1$ and
let $(i_1,\ldots,i_{j-1})$ be the unique valid transcript of length $j-1$.
Now for every other sequence $(i'_1,\ldots,i'_{j-1}) \ne 
(i_1,\ldots,i_{j-1})$ at least one of 
$\tildeptranscript{A}(i_1,\ldots,i_{j-1})$ or
$\tildeptranscript{B}(i_1,\ldots,i_{j-1})$ is zero. Suppose 
$\tildeptranscript{A}(i_1,\ldots,i_{j-1}) = 0$.
Then by the previous paragraph we have both 
$\tildeptranscript{A}(i_1,\ldots,i_{j-1},0) = 0$
and 
$\tildeptranscript{A}(i_1,\ldots,i_{j-1},1) = 0$.
So none of the ``children'' of invalid transcripts are valid. 
We turn to the unique valid sequence of length $j-1$. Suppose $j$ is an even
index (and so Bob speaks in round $j$). Since 
$\tildeptranscript{A}(i_1,\ldots,i_{j-1}) = 1$ we must have 
$\tildeptranscript{A}(i_1,\ldots,i_{j-1},0) = 
\tildeptranscript{A}(i_1,\ldots,i_{j-1},1) = 1$ (since there exists a Bob
strategy making each possible transcript valid). 
On the other hand there must exist a unique Bob message $b = f_B^{j-1}(i_1,ldots,i_{j-1}) \in \{0,1\}$
in round $j$ given the transcript $(i_1,\ldots,i_{j-1})$ thus far.
For this $b$ we have 
$\tildeptranscript{A}(i_1,\ldots,i_{j-1},b) = 1$ 
and 
$\tildeptranscript{A}(i_1,\ldots,i_{j-1},1-b) = 0$, making
$(i_1,\ldots,i_{j-1},b)$ the unique valid transcript at level $j$.

Finally we prove the converse showing that every vector $\strat{A} \in S_A$
corresponds to a strategy $f_A$. The key step here is to show that there
exists a vector $\tildestrat{A} \in \tilde{S}_A$ such that 
$\strat{A} = \tildestrat{A} * v$. The strategy $f_A$ can then be read off
from $\tildeptranscript{A} =
\pt_A(\tildestrat{A})$.
By definition of membership in $\strat{A} \in S_A \subseteq K_A$ we have
that there exists 
$\tilde{\xi'} \in \tilde{K}_A$ such that 
$\strat{A} = \tilde{\xi'} * v$. 
Let $\tilde{\chi}_{A,\tilde{\xi'}} = \pt_A(\tilde{\xi'})$. 
We show how to use this to create a
Boolean vector with the same property. We first define a function
$\chi_1(i_1,\ldots,i_j) \in \{0,1,?\}$ as follows:
We define $\chi_1(i_1,\ldots,i_k) = \strat{A}(i_1,\ldots,i_k)$ if
$v(i_1,\ldots,i_k) = 1$ and 
$\chi_1(i_1,\ldots,i_k) = ?$ if 
$v(i_1,\ldots,i_k) = 0$.
For $j$ going down from $k-1$ to $0$ we proceed as follows:
If $j$ is even (and it is Alice's turn to speak), then we set
$\chi_1(i_1,\ldots,i_j) = 1$ if 
$\chi_1(i_1,\ldots,i_j, b) = 1$ for some $b \in \{0,1\}$,
we set $\chi_1(i_1,\ldots,i_j) = 0$ if 
$\chi_1(i_1,\ldots,i_j, b) = 0$ for every $b \in \{0,1\}$,
and we set 
$\chi_1(i_1,\ldots,i_j) = ?$ otherwise.
If $j$ is odd (and it is Bob's turn to speak), then we set
$\chi_1(i_1,\ldots,i_j) = 1$ if 
$\chi_1(i_1,\ldots,i_j, b) = 1$ for some $b \in \{0,1\}$,
else we set $\chi_1(i_1,\ldots,i_j) = 0$ if 
$\chi_1(i_1,\ldots,i_j, b) = 0$ for some $b \in \{0,1\}$,
and we set 
$\chi_1(i_1,\ldots,i_j) = ?$ otherwise.
We now assert, by downward induction on $j$ that
if $\chi_1(i_1,\ldots,i_j) \in \{0,1\}$ then
$\chi_1(i_1,\ldots,i_j) = 
\tildeptranscript{A,\tilde{\xi'}}(i_1,\ldots,i_j)$. 
This is true trivially for $j=k$. 
For odd $j < k$, if 
$\chi_1(i_1,\ldots,i_j, b) \in \{0,1\}$ for some $b \in \{0,1\}$,
then 
$\tildeptranscript{A,\tilde{\xi'}}(i_1,\ldots,i_j)
= \tildeptranscript{A,\tilde{\xi'}}(i_1,\ldots,i_j,1-b)
= \tildeptranscript{A,\tilde{\xi'}}(i_1,\ldots,i_j,b)
\chi_1(i_1,\ldots,i_j, b)$ by induction
and so we have 
$\tildeptranscript{A,\tilde{\xi'}}(i_1,\ldots,i_j)
= \chi_1(i_1,\ldots,i_j)$. In all other cases
$\chi_1(i_1,\ldots,i_j)= ?$. Note also that we always have
$\chi_1(i_1,\ldots,i_j,b) =
\chi_1(i_1,\ldots,i_j,1-b)$ when both are in $\{0,1\}$.
For even $j < k$, we reason similarly. We first note that
both 
$\chi_1(i_1,\ldots,i_j,0)$ and 
$\chi_1(i_1,\ldots,i_j,1)$ can't be $1$, since then we would have
$\tildeptranscript{A,\tilde{\xi'}}(i_1,\ldots,i_j,0)
= \tildeptranscript{A,\tilde{\xi'}}(i_1,\ldots,i_j,1) = 1$
and this would imply 
$\tildeptranscript{A,\tilde{\xi'}}(i_1,\ldots,i_j) = 2$.
We conclude that if
$\chi_1(i_1,\ldots,i_j)=1$ then 
$\chi_1(i_1,\ldots,i_j,b)=1$ for exactly one $b \in \{0,1\}$
and so 
$\tildeptranscript{A,\tilde{\xi'}}(i_1,\ldots,i_j,b) = 1$ for at least 
one $b$ as well in which case we would again have 
$\tildeptranscript{A,\tilde{\xi'}}(i_1,\ldots,i_j) = 1$.
The case of 
$\chi_1(i_1,\ldots,i_j)=0$  is simpler since this arises when
$\chi_1(i_1,\ldots,i_j,0)=
\chi_1(i_1,\ldots,i_j,1)=0$ and in this case we have 
$\tildeptranscript{A,\tilde{\xi'}}(i_1,\ldots,i_j,0)
= \tildeptranscript{A,\tilde{\xi'}}(i_1,\ldots,i_j,1) = 0$
and so $\tildeptranscript{A,\tilde{\xi'}}(i_1,\ldots,i_j) = 0$ as well.
Finally we note that 
$\chi_1(i_1,\ldots,i_j)=?$ only if both 
$\chi_1(i_1,\ldots,i_j,1)=?$ and 
$\chi_1(i_1,\ldots,i_j,0)=?$ or $j$ is even and 
$\chi_1(i_1,\ldots,i_j,b)=0$ for some $b \in \{0,1\}$ and
$\chi_1(i_1,\ldots,i_j,1-b)=?$.

We now use these properties of $\chi_1$ to ``complete'' to a
$\{0,1\}$-valued function $\chi_2$ as follows. Set $\chi_2() = 1$
and now for $j$ increasing from $0$ to $k-1$  proceed as follows:
If $\chi_1(i_1,\ldots,i_j,b) \in \{0,1\}$ then let 
$\chi_2(i_1,\ldots,i_j,b) = 
\chi_2(i_1,\ldots,i_j,b)$.
For the remaining choices of $(i_1,\ldots,i_j,b)$ assign
$\chi_2(i_1,\ldots,i_j,b)$ as follows:
If $j$ is odd then set $\chi_2(i_1,\ldots,i_j,b) = 
\chi_2(i_1,\ldots,i_j)$.
If $j$ is even and $\chi_1(i_1,\ldots,i_j,1-b) \in \{0,1\}$
then set 
$\chi_2(i_1,\ldots,i_j,b)
= \chi_2(i_1,\ldots,i_j)
- \chi_2(i_1,\ldots,i_j,1-b)$.
Else if we have 
$\chi_1(i_1,\ldots,i_j,b) = ?$ for both $b \in \{0.1\}$ then
set 
$\chi_2(i_1,\ldots,i_j,0) = 0$ (arbitrarily) and
$\chi_2(i_1,\ldots,i_j,1) = \chi_2(i_1,\ldots,i_j)$.

Now let $\tildestrat{A}(i_1,\ldots,i_k) = \chi_2(i_1,\ldots,i_k)$. It can be
verified that $\tildeptranscript{A} = \pt_A(\tildestrat{A})$ 
satisfies $\tildeptranscript{A}(i_1,\ldots,i_j) = 
\chi_2(i_1,\ldots,i_j)$ for every $i_1,\ldots,i_j$ and thus satisfies
the conditions of membership in $\tilde{S}_A$.

Now to derive the strategy $f_A$, for even $j$,
if $\tildeptranscript{A}(i_1,\ldots,i_j) = 1$ 
we set $f_A^j(i_1,\ldots,i_j) = b$
such that $\tildeptranscript{A}(i_1,\ldots,i_j,b) = 1$ 
(such a $b$ must exist); and
set 
$f_A^j(i_1,\ldots,i_j) = 1$ (arbitrarily) otherwise.
It can be verified that for this strategy the vector representation leads exactly
to vectors $\tildestrat{A} \in \tilde{S}_A$ and $\strat{A} = \tildestrat{A} * v
\in S_A$.
\end{proof}

More significantly for us, the above equivalence also holds for probabilistic communication (i.e., with private randomness), {as defined in the third paragraph of~\autoref{ssec:strategy:preliminaries}. Recall then that a (private-coin) probabilistic strategy is defined by two sequences $(f^{2i}_A)_{1\leq i < k/2}$, $(f^{2i+1}_B)_{1\leq i < k/2}$ of functions taking values in $[0,1]$, and a verdict function~$v$}. Here the fact that the set of strategies forms a convex space is important to us.

\begin{proposition}\label{prop:ka:kb}
$K_A$ and $K_B$ are closed convex sets that correspond to the set of probabilistic communication (and decision)
strategies with $k$ bits. More precisely, for every probabilistic strategy $f_A$ of Alice there exists a vector
$\tildestrat{A} \in \tilde{K}_A$ and $\strat{A} \in K_A$ and for every strategy $f_B$ of Bob there exists a vector
$\tildestrat{B} \in \tilde{K}_B$ and $\strat{B}\in K_B$ such that $\tildestrat{A}(\ell)\cdot \tildestrat{B}(\ell)$ is the probability that
$\ell \in \{0,1\}^k$ is the transcript of the interaction between Alice and Bob under strategies
$f_A$ and $f_B$ and $\ip{\strat{A}}{\strat{B}}$ is the acceptance probability of the interaction.
Conversely every vector $\strat{A} \in K_A$ corresponds to a probabilistic strategy $f_A$ for Alice
(and similarly for Bob, with $\ip{\strat{A}}{\strat{B}}$ being the acceptance probability of the  protocol).
\end{proposition}

\begin{proof}
The fact that $K_A$ and $K_B$ are closed and convex sets is straightforward from their definition.

We first show that strategies $f_A$ and $f_B$ can be converted into vectors in
$\tilde{K}_A$ and $\tilde{K}_B$ respectively. 
We define $\tildeptranscript{A}$ inductively as follows. Let
$\tildeptranscript{A}() = 1$.
Further let 
$\tildeptranscript{A}(i_1,\ldots,i_j,1) = 
\tildeptranscript{A}(i_1,\ldots,i_j,0) = 
\tildeptranscript{A}(i_1,\ldots,i_j)$ if $j$ is $j$ is odd.
Finally let 
$\tildeptranscript{A}(i_1,\ldots,i_j,1) = 
\tildeptranscript{A}(i_1,\ldots,i_j) \cdot f_A^{j}(i_1,\ldots,i_j)$ 
and 
$\tildeptranscript{A}(i_1,\ldots,i_j,1) = 
\tildeptranscript{A}(i_1,\ldots,i_j) \cdot (1- f_A^{j}(i_1,\ldots,i_j))$ if
$j$ is even.
Now let $\tildestrat{A}(i_1,\ldots,i_k) = 
\tildeptranscript{A}(i_1,\ldots,i_k)$.
It can be verified that $\tildestrat{A} \in \tilde{K}_A$ and indeed the vector
$\ptranscript{A} = \pt_A(\tildestrat{A})$ (as given by~\autoref{def:xa:xb}) is the vector $\tildeptranscript{A}$. Taking $\strat{A} = \tildestrat{A} * v$ gives us
the vector $\strat{A} \in K_A$ as needed. {(The definition of $\strat{B}\in K_B$ is similar.)}

Now for the more important direction, we claim that every vector $\strat{A} \in
K_A$ corresponds to a strategy $f_A$. Note by definition that since
$\strat{A} \in K_A$, there exists $\tildestrat{A} \in \tilde{K}_A$ such that
$\strat{A} = \tildestrat{A} * v$. Let $\tildeptranscript{A}=\pt_A(\tildestrat{A})$ be the vector
obtained from $\tildestrat{A}$. For even $j$ and $i_1,\ldots,i_j \in \{0,1\}$. let
$f_A^{j}(i_1,\ldots,i_j) = 
\tildeptranscript{A}(i_1,\ldots,i_j,1)
/\tildeptranscript{A}(i_1,\ldots,i_j)$ for even $j$. (If 
$\tildeptranscript{A}(i_1,\ldots,i_j) = 0$, we define 
$f_A^{j}(i_1,\ldots,i_j) = 1$
Note that since 
$\tildeptranscript{A}(i_1,\ldots,i_t) \in [0,1]$ for every $t$ and
$\tildeptranscript{A}(i_1,\ldots,i_j) = 
\tildeptranscript{A}(i_1,\ldots,i_j,0) + 
\tildestrat{A}(i_1,\ldots,i_j,1)$, we have that 
$f_A^{j}(i_1,\ldots,i_j) \in [0,1]$ and so $f_A$ represents a strategy for Alice.
It can further be verified that if we apply the transformation of the previous
paragraph to this strategy $f_A$, we recover $\tildestrat{A}$ and so this
correspondence is indeed bidirectional.

Finally we verify that the acceptance probabilities are given by the inner
product function. Consider strategies given by vectors 
$\strat{A} \in K_A$ and $\strat{B} \in K_B$ with corresponding vectors 
$\tildestrat{A} \in \tilde{K}_A$ and $\tildestrat{B} \in \tilde{K}_B$ such that
$\strat{A} = \tildestrat{A} * v$
and $\strat{B} = \tildestrat{B} * v$. 
{Let $\tildeptranscript{A}=\pt_A(\tildestrat{A})$ and $\tildeptranscript{B}=\pt_B(\tildestrat{B})$; and} 
let $f_A$ and $f_B$ be the probabilistic strategies corresponding to $\strat{A}$
and $\strat{B}$ respectively as given by the previous paragraph. 
We claim, by induction on $j$, that for every $i_1,\ldots,i_j$ the probability
that the interaction reaches the partial transcript $i_1,\ldots,i_j$ under
strategies $f_A$, $f_B$ is $\tildeptranscript{A}(i_1,\ldots,i_j)\cdot
\tildeptranscript{B}(i_1,\ldots,i_j)$.
This is certainly true for $j=0$ where 
$\tildeptranscript{A}() = \tildeptranscript{B}() = 1$.
Now consider a partial transcript $i_1,\ldots,i_j$. 
By induction the probability of reaching this transcript is 
$p = \tildeptranscript{A}(i_1,\ldots,i_j) \cdot 
\tildeptranscript{B}(i_1,\ldots,i_j)$.
Suppose it is Alice's turn to speak. Then
$\tildeptranscript{B}(i_1,\ldots,i_j,0) = 
\tildeptranscript{B}(i_1,\ldots,i_j,1) = 
\tildeptranscript{B}(i_1,\ldots,i_j)$.
And 
$\tildeptranscript{A}(i_1,\ldots,i_j,0) +
\tildeptranscript{A}(i_1,\ldots,i_j,1) = 
\tildeptranscript{A}(i_1,\ldots,i_j)$.
Thus the probability of reaching the partial transcript
$i_1,\ldots,i_j,1$ is
$p \cdot f_A^{j}(i_1,\ldots,i_j) = 
\tildeptranscript{A}(i_1,\ldots,i_j) \cdot 
\tildeptranscript{B}(i_1,\ldots,i_j) \cdot 
\tildeptranscript{A}(i_1,\ldots,i_j,1)
/\tildeptranscript{A}(i_1,\ldots,i_j)
= 
\tildeptranscript{A}(i_1,\ldots,i_j,1) \cdot
\tildeptranscript{B}(i_1,\ldots,i_j)
= 
\tildeptranscript{A}(i_1,\ldots,i_j,1) \cdot
\tildeptranscript{B}(i_1,\ldots,i_j,1)$.
The calculation for the extension $i_1,\ldots,i_j,0$ is similar
and uses the fact that 
$\tildeptranscript{A}(i_1,\ldots,i_j) =
\tildeptranscript{A}(i_1,\ldots,i_j,0) +
\tildeptranscript{A}(i_1,\ldots,i_j,1)$, and
the probability that the transcript extends to
$i_1,\ldots,i_j,0$ conditioned on being at
$i_1,\ldots,i_j$ is $1 - f_A^j(i_1,\ldots,i_j)$.

From the above, we conclude that the probability of reaching a final transcript 
$\ell$ is 
$\tildeptranscript{A}(\ell)
\cdot \tildeptranscript{B}(\ell)
= 
\tildestrat{A}(\ell)\cdot
\tildestrat{B}(\ell)$. 
Thus the acceptance probability equals $\sum_{\ell: v(\ell) = 1}
\tildestrat{A}(\ell)\cdot
\tildestrat{B}(\ell)
= \sum_{\ell\in \{0,1\}^k} 
\tildestrat{A}(\ell)\cdot
\tildestrat{B}(\ell) \cdot v(\ell)^2 
= \sum_{\ell\in \{0,1\}^k} 
\strat{A}(\ell)\cdot \strat{B}(\ell)
= \langle 
\strat{A},\strat{B}\rangle$.
\end{proof}
\else
\begin{definition}\label{def:xa:xb}
For vector $\strat{} \in [0,1]^{2^k}$ and $i_1,\ldots,i_j \in \{0,1\}$
let $\ptranscript{A,\strat{}}(i_1,\ldots,i_j) = 
\ptranscript{A}(i_1,\ldots,i_j) \in [0,1]$ and 
$\ptranscript{B,\strat{}}(i_1,\ldots,i_j) = 
\ptranscript{B}(i_1,\ldots,i_j) \in [0,1]$ be defined as
follows:
\begin{align*}
 \ptranscript{A}(i_1,\ldots,i_j) & = \begin{cases}
				\strat{}(i_1,\ldots,i_k) &\text{ if } j=k\\
				\ptranscript{A}(i_1,\ldots,i_j,0) + \ptranscript{A}(i_1,\ldots,i_j,1) &\text{ if } j \text{ is even.}\\
				\frac12\left(\ptranscript{A}(i_1,\ldots,i_j,0) + \ptranscript{A}(i_1,\ldots,i_j,1)\right) &\text{ if } j \text{ is odd.}\\
			\end{cases}
\end{align*}
\begin{align*}
 \ptranscript{B}(i_1,\ldots,i_j) & = \begin{cases}
				\strat{}(i_1,\ldots,i_k) &\text{ if } j=k\\
				\frac12\left(\ptranscript{B}(i_1,\ldots,i_j,0) + \ptranscript{B}(i_1,\ldots,i_j,1)\right) &\text{ if } j \text{ is even.}\\
				\ptranscript{B}(i_1,\ldots,i_j,0) + \ptranscript{B}(i_1,\ldots,i_j,1) &\text{ if } j \text{ is odd.}\\
			\end{cases}
\end{align*}

Define
\[\tilde{K}_A = \tilde{K}^{(k)}_A = \setOfSuchThat{ \strat{} \in [0,1]^{2^k} }{ \ptranscript{A}() = 1 \text{ and }\forall \text{ odd
}j, \forall i_1,\ldots,i_j \in \{0,1\}, ~~
				\ptranscript{A}(i_1,\ldots,i_j,0) =
				\ptranscript{A}(i_1,\ldots,i_j,1) },
				\]
and
\[\tilde{K}_B = \tilde{K}^{(k)}_B = \setOfSuchThat{\strat{} \in [0,1]^{2^k} }{ \ptranscript{B}() = 1 \text{ and }\forall \text{
even }j, \forall i_1,\ldots,i_j \in \{0,1\}, ~~
				\ptranscript{B}(i_1,\ldots,i_j,0) =
				\ptranscript{B}(i_1,\ldots,i_j,1) }.
				\]
Let $K_A = \setOfSuchThat{\strat{} \ast v }{ \strat{} \in \tilde{K}_A }$, where $v \in \{0,1\}^{2^k}$ is
given by $v(i_1,\ldots,i_k) = i_k$ (and $a \ast b$ denotes coordinate-wise
multiplication of vectors $a$ and $b$).
Let
$\tilde{S}_A = \tilde{K}_A \cap \{0,1\}^{2^k}$,
$\tilde{S}_B = \tilde{K}_B \cap \{0,1\}^{2^k}$,
$S_A = K_A \cap \{0,1\}^{2^k}$,
and $S_B = K_B \cap \{0,1\}^{2^k}$.
\end{definition}

\begin{figure}[h]\centering
\scalebox{0.8}{
    \begin{tikzpicture}[->,>=stealth',
      level 1/.style={sibling distance=85mm},
      level 2/.style={sibling distance=40mm},
      level 3/.style={sibling distance=0mm},
      level distance=20mm,
      ]
    \node {$\underbrace{\ptranscript{A}(i_1,\dots,i_{2j})}_{\color{black!65}0+1=1}$}
        child{ node {$\underbrace{\ptranscript{A}(i_1,\dots,i_{2j},i_{2j+1}=0)}_{\color{black!65}\frac{1}{2}(0+0)=0}$}
                child{ node {$\underbrace{\ptranscript{A}(i_1,\dots,i_{2j},i_{2j+1},0)}_{\color{black!65}0}$} edge from parent node[above left] {$0$}
                }
                child{ node {$\underbrace{\ptranscript{A}(i_1,\dots,i_{2j},i_{2j+1},1)}_{\color{black!65}0}$} edge from parent node[above right] {$1$}
                }
           edge from parent node[above left] {$0$}
        }
        child{ node {$\underbrace{\ptranscript{A}(i_1,\dots,i_{2j},i_{2j+1}=1)}_{\color{black!65}\frac{1}{2}(1+1)=1}$}
                child{ node {$\underbrace{\ptranscript{A}(i_1,\dots,i_{2j},i_{2j+1},0)}_{\color{black!65}1}$} edge from parent node[above left] {$0$}
                }
                child{ node {$\underbrace{\ptranscript{A}(i_1,\dots,i_{2j},i_{2j+1},1)}_{\color{black!65}1}$} edge from parent node[above right] {$1$}
                }
          edge from parent node[above right] {$1$}
		    }
    ;
    \end{tikzpicture}
}
\caption{\label{fig:constraints:xa} Illustration of the constraints on $\ptranscript{A}$ (\autoref{def:xa:xb}).}
\end{figure}

In what follows we first focus on deterministic communication strategies
and show that $\tilde{S}_A$, $\tilde{S}_B$ correspond to
the space of deterministic communication strategies for Alice and Bob,
while $S_A$ and $S_B$ correspond to outputs computed by such strategies.
This step is not strictly needed for this paper since our main focus is
on probabilistic strategies and the convex sets $K_A$ and $K_B$, but
we include it for completeness.

\begin{proposition}
\label{prop:sa:sb}.
$\tilde{S}_A$ and $\tilde{S}_B$ correspond to the set of deterministic communication
strategies with $k$ bits. For every strategy $f_A$ of Alice there exist
vectors
$\tildestrat{A} \in \tilde{S}_A$ and $\strat{A}\in S_A$ and for every strategy $f_B$ of Bob there exist  vectors $\tildestrat{B} \in
\tilde{S}_B$ and $\strat{B} \in S_B$
such that if $\ell \in \{0,1\}^k$ is the transcript of the interaction between Alice and Bob
under strategies $f_A$ and $f_B$, then $\ell$ is the unique sequence satisfying 
$\tildestrat{A}(\ell) = \tildestrat{B}(\ell) = 1$ and $\ip{\strat{A}}{\strat{B}} = 1$ if the interaction
accepts and $0$ otherwise.

Conversely every vector $\strat{A} \in S_A$ corresponds to a strategy $f_A$ for Alice
(and similarly for Bob) such that Alice and Bob accept the interaction iff $\ip{\strat{A}}{\strat{B}} = 1$.
\end{proposition}
\begin{proof}
Given $f_A$ to construct $\tildestrat{A}$, we let $\tildestrat{A}(\ell) = 1$ if there exists $f_{B,\ell}$ such that the final
transcript of the interaction given by $f_A$ and $f_{B,\ell}$ is $\ell$.
Furthermore let $\tildeptranscript{A}(i_1,\ldots,i_j) = 1$ if there exists a Bob
strategy $f_{B,i_1,\ldots,i_j}$ such that $i_1,\ldots,i_j$ is the partial
transcript of the interaction between Alice and Bob; {otherwise let $\tildeptranscript{A}(i_1,\ldots,i_j) = 0$}.
It is now straightforward to verify that the $\tildeptranscript{A}(i_1,\ldots,i_j)$ satisfy the conditions of the definition of $\tildestrat{A}$ and the conditions required for
membership in $\tilde{K}_A$. In particular we have the following three conditions: (1)
$\tildeptranscript{A}() = 1$ since the empty transcript is a legal partial transcript.  (2) If $j$ is an
even index (and so Alice speaks in round $j+1$) and $\ptranscript{A}(i_1,\ldots,i_j) = 0$ (so the partial
transcript $i_1,\ldots,i_j$ is not reachable given Alice's strategy), then we must have
$\ptranscript{A}(i_1,\ldots,i_j,0) = \ptranscript{A}(i_1,\ldots,i_j,1) = 0$ (no extension is reachable either). If
$\ptranscript{A}(i_1,\ldots,i_j) = 1$ then exactly one of the extensions must be reachable (based on Alice's
message at this stage) and so again we have
$\ptranscript{A}(i_1,\ldots,i_j) = \ptranscript{A}(i_1,\ldots,i_j,0) + \ptranscript{A}(i_1,\ldots,i_j,1)$. (3) If $j$ is odd and it is
Bob's turn to speak, then again if $\ptranscript{A}(i_1.\ldots,i_j) = 0$ we have
$\ptranscript{A}(i_1,\ldots,i_j,0) = \ptranscript{A}(i_1,\ldots,i_j,1) = 0$. On the other hand if
$\ptranscript{A}(i_1,\ldots,i_j) = 1$ then for each extension there exists a strategy of Bob that permits this
extension and so we have
$\ptranscript{A}(i_1,\ldots,i_j,0) = \ptranscript{A}(i_1,\ldots,i_j,1) = 1$ satisfying the condition for odd $j$.
The above three conditions verify membership in $\tilde{K}_A$ and since $\tildestrat{A}$ is a 0/1 vector, we also have
$\tildestrat{A} \in \tilde{S}_A$. The vector $\strat{A} = \tildestrat{A} \ast v$ gives the
corresponding vector in $S_A$.

Given a pair of strategies $f_A$ and $f_B$, let 
$\tildestrat{A}$ and $\tildestrat{B}$ be the corresponding vectors representing the strategies, and let 
$\tildeptranscript{A}(i_1,\ldots,i_j)$ 
and 
$\tildeptranscript{B}(i_1,\ldots,i_j)$ denote the partial transcripts.
We now prove the existence and uniqueness of a leaf $\ell$ such that 
$\tildestrat{A}(\ell) = \tildestrat{B}(\ell) = 1$.
We do so by showing, by induction on $j$, that for every $j \in \{0,\ldots,k\}$
there exists a unique sequence $(i_1,\ldots,i_j)$ such
that $\tildeptranscript{A}(i_1,\ldots,i_j) = \tildeptranscript{B}(i_1,\ldots,i_j) =
1$. Using the fact that 
$\tildeptranscript{A}(i_1,\ldots,i_k) = 
\tildestrat{A}(i_1,\ldots,i_k)$,
we get the desired existence and uniqueness (for $\ell = (i_1,\ldots,i_k)$).
We refer to a sequence $(i_1,\ldots,i_j)$ as a valid (partial) transcript if
$\tildeptranscript{A}(i_1,\ldots,i_j) = \tildeptranscript{B}(i_1,\ldots,i_j) = 1$
and as an invalid transcript otherwise.
The base case of the induction is true since there is only one transcript of
length $0$ and 
we have $\tildeptranscript{A}() = \tildeptranscript{B}() = 1$ so the empty
transcript is the unique valid transcript.
Now assume the statement is true for transcripts of length $j-1$ and
let $(i_1,\ldots,i_{j-1})$ be the unique valid transcript of length $j-1$.
Now for every other sequence $(i'_1,\ldots,i'_{j-1}) \ne 
(i_1,\ldots,i_{j-1})$ at least one of 
$\tildeptranscript{A}(i_1,\ldots,i_{j-1})$ or
$\tildeptranscript{B}(i_1,\ldots,i_{j-1})$ is zero. Suppose 
$\tildeptranscript{A}(i_1,\ldots,i_{j-1}) = 0$.
Then by the previous paragraph we have both 
$\tildeptranscript{A}(i_1,\ldots,i_{j-1},0) = 0$
and 
$\tildeptranscript{A}(i_1,\ldots,i_{j-1},1) = 0$.
So none of the ``children'' of invalid transcripts are valid. 
We turn to the unique valid sequence of length $j-1$. Suppose $j$ is an even
index (and so Bob speaks in round $j$). Since 
$\tildeptranscript{A}(i_1,\ldots,i_{j-1}) = 1$ we must have 
$\tildeptranscript{A}(i_1,\ldots,i_{j-1},0) = 
\tildeptranscript{A}(i_1,\ldots,i_{j-1},1) = 1$ (since there exists a Bob
strategy making each possible transcript valid). 
On the other hand there must exist a unique Bob message $b = f_B^{j-1}(i_1,ldots,i_{j-1}) \in \{0,1\}$
in round $j$ given the transcript $(i_1,\ldots,i_{j-1})$ thus far.
For this $b$ we have 
$\tildeptranscript{A}(i_1,\ldots,i_{j-1},b) = 1$ 
and 
$\tildeptranscript{A}(i_1,\ldots,i_{j-1},1-b) = 0$, making
$(i_1,\ldots,i_{j-1},b)$ the unique valid transcript at level $j$.

Finally we prove the converse showing that every vector $\strat{A} \in S_A$
corresponds to a strategy $f_A$. The key step here is to show that there
exists a vector $\tildestrat{A} \in \tilde{S}_A$ such that 
$\strat{A} = \tildestrat{A} * v$. The strategy $f_A$ can then be read off
from $\tildeptranscript{A} \eqdef \ptranscript{A.\tildestrat{A}}$.
By definition of membership in $\strat{A} \in S_A \subseteq K_A$ we have
that there exists 
$\tilde{\xi'} \in \tilde{K}_A$ such that 
$\strat{A} = \tilde{\xi'} * v$. We show how to use this to create a
Boolean vector with the same property. We first define a function
$\chi_1(i_1,\ldots,i_j) \in \{0,1,?\}$ as follows:
We define $\chi_1(i_1,\ldots,i_k) = \strat{A}(i_1,\ldots,i_k)$ if
$v(i_1,\ldots,i_k) = 1$ and 
$\chi_1(i_1,\ldots,i_k) = ?$ if 
$v(i_1,\ldots,i_k) = 0$.
For $j$ going down from $k-1$ to $0$ we proceed as follows:
If $j$ is even (and it is Alice's turn to speak), then we set
$\chi_1(i_1,\ldots,i_j) = 1$ if 
$\chi_1(i_1,\ldots,i_j, b) = 1$ for some $b \in \{0,1\}$,
we set $\chi_1(i_1,\ldots,i_j) = 0$ if 
$\chi_1(i_1,\ldots,i_j, b) = 0$ for every $b \in \{0,1\}$,
and we set 
$\chi_1(i_1,\ldots,i_j) = ?$ otherwise.
If $j$ is odd (and it is Bob's turn to speak), then we set
$\chi_1(i_1,\ldots,i_j) = 1$ if 
$\chi_1(i_1,\ldots,i_j, b) = 1$ for some $b \in \{0,1\}$,
else we set $\chi_1(i_1,\ldots,i_j) = 0$ if 
$\chi_1(i_1,\ldots,i_j, b) = 0$ for some $b \in \{0,1\}$,
and we set 
$\chi_1(i_1,\ldots,i_j) = ?$ otherwise.
We now assert, by downward induction on $j$ that
if $\chi_1(i_1,\ldots,i_j) \in \{0,1\}$ then
$\chi_1(i_1,\ldots,i_j) = 
\tildeptranscript{A,\tilde{\xi'}}(i_1,\ldots,i_j)$. 
This is true trivially for $j=k$. 
For odd $j < k$, if 
$\chi_1(i_1,\ldots,i_j, b) \in \{0,1\}$ for some $b \in \{0,1\}$,
then 
$\tildeptranscript{A,\tilde{\xi'}}(i_1,\ldots,i_j)
= \tildeptranscript{A,\tilde{\xi'}}(i_1,\ldots,i_j,1-b)
= \tildeptranscript{A,\tilde{\xi'}}(i_1,\ldots,i_j,b)
\chi_1(i_1,\ldots,i_j, b)$ by induction
and so we have 
$\tildeptranscript{A,\tilde{\xi'}}(i_1,\ldots,i_j)
= \chi_1(i_1,\ldots,i_j)$. In all other cases
$\chi_1(i_1,\ldots,i_j)= ?$. Note also that we always have
$\chi_1(i_1,\ldots,i_j,b) =
\chi_1(i_1,\ldots,i_j,1-b)$ when both are in $\{0,1\}$.
For even $j < k$, we reason similarly. We first note that
both 
$\chi_1(i_1,\ldots,i_j,0)$ and 
$\chi_1(i_1,\ldots,i_j,1)$ can't be $1$, since then we would have
$\tildeptranscript{A,\tilde{\xi'}}(i_1,\ldots,i_j,0)
= \tildeptranscript{A,\tilde{\xi'}}(i_1,\ldots,i_j,1) = 1$
and this would imply 
$\tildeptranscript{A,\tilde{\xi'}}(i_1,\ldots,i_j) = 2$.
We conclude that if
$\chi_1(i_1,\ldots,i_j)=1$ then 
$\chi_1(i_1,\ldots,i_j,b)=1$ for exactly one $b \in \{0,1\}$
and so 
$\tildeptranscript{A,\tilde{\xi'}}(i_1,\ldots,i_j,b) = 1$ for at least 
one $b$ as well in which case we would again have 
$\tildeptranscript{A,\tilde{\xi'}}(i_1,\ldots,i_j) = 1$.
The case of 
$\chi_1(i_1,\ldots,i_j)=0$  is simpler since this arises when
$\chi_1(i_1,\ldots,i_j,0)=
\chi_1(i_1,\ldots,i_j,1)=0$ and in this case we have 
$\tildeptranscript{A,\tilde{\xi'}}(i_1,\ldots,i_j,0)
= \tildeptranscript{A,\tilde{\xi'}}(i_1,\ldots,i_j,1) = 0$
and so $\tildeptranscript{A,\tilde{\xi'}}(i_1,\ldots,i_j) = 0$ as well.
Finally we note that 
$\chi_1(i_1,\ldots,i_j)=?$ only if both 
$\chi_1(i_1,\ldots,i_j,1)=?$ and 
$\chi_1(i_1,\ldots,i_j,0)=?$ or $j$ is even and 
$\chi_1(i_1,\ldots,i_j,b)=0$ for some $b \in \{0,1\}$ and
$\chi_1(i_1,\ldots,i_j,1-b)=?$.

We now use these properties of $\chi_1$ to ``complete'' to a
$\{0,1\}$-valued function $\chi_2$ as follows. Set $\chi_2() = 1$
and now for $j$ increasing from $0$ to $k-1$  proceed as follows:
If $\chi_1(i_1,\ldots,i_j,b) \in \{0,1\}$ then let 
$\chi_2(i_1,\ldots,i_j,b) = 
\chi_2(i_1,\ldots,i_j,b)$.
For the remaining choices of $(i_1,\ldots,i_j,b)$ assign
$\chi_2(i_1,\ldots,i_j,b)$ as follows:
If $j$ is odd then set $\chi_2(i_1,\ldots,i_j,b) = 
\chi_2(i_1,\ldots,i_j)$.
If $j$ is even and $\chi_1(i_1,\ldots,i_j,1-b) \in \{0,1\}$
then set 
$\chi_2(i_1,\ldots,i_j,b)
= \chi_2(i_1,\ldots,i_j)
- \chi_2(i_1,\ldots,i_j,1-b)$.
Else if we have 
$\chi_1(i_1,\ldots,i_j,b) = ?$ for both $b \in \{0.1\}$ then
set 
$\chi_2(i_1,\ldots,i_j,0) = 0$ (arbitrarily) and
$\chi_2(i_1,\ldots,i_j,1) = \chi_2(i_1,\ldots,i_j)$.

Now let $\tildestrat{A}(i_1,\ldots,i_k) = \chi_2(i_1,\ldots,i_k)$. It can be
verified that $\tildeptranscript{A} = 
\ptranscript{A,\tildestrat{A}}$ satisfies 
$\tildeptranscript{A}(i_1,\ldots,i_j) = 
\chi_2(i_1,\ldots,i_j)$ for every $i_1,\ldots,i_j$ and thus satisfies
the conditions of membership in $\tilde{S}_A$. 

Now to derive the strategy $f_A$, for even $j$,
if $\tildeptranscript{A}(i_1,\ldots,i_j) = 1$ 
we set $f_A^j(i_1,\ldots,i_j) = b$
such that $\tildeptranscript{A}(i_1,\ldots,i_j,b) = 1$ 
(such a $b$ must exist); and
set 
$f_A^j(i_1,\ldots,i_j) = 1$ (arbitrarily) otherwise.
It can be verified that for this strategy the vector representation leads exactly
to vectors $\tildestrat{A} \in \tilde{S}_A$ and $\strat{A} = \tildestrat{A} * v
\in S_A$.
\end{proof}

More significantly for us, the above equivalence also holds for probabilistic communication (i.e., with
private randomness), {as defined in the third paragraph of~\autoref{ssec:strategy:preliminaries}. Recall then that a (private-coin) probabilistic strategy is defined by two sequences $(f^{2i}_A)_{1\leq i < k/2}$, $(f^{2i+1}_B)_{1\leq i < k/2}$ of functions taking values in $[0,1]$, and a verdict function~$v$}. Here the fact that the set of strategies forms a convex space is important to us.

\begin{proposition}\label{prop:ka:kb}
$K_A$ and $K_B$ are {closed} convex sets
that correspond to the set of probabilistic communication (and
decision)
strategies with $k$ bits. More precisely, for every probabilistic strategy $f_A$ of Alice there exists a vector
$\tildestrat{A} \in \tilde{K}_A$ and $\strat{A} \in K_A$
and for every strategy $f_B$ of Bob there exists a vector
$\tildestrat{B} \in \tilde{K}_B$
and $\strat{B}\in K_B$
such that $\tildestrat{A}(\ell)\cdot \tildestrat{B}(\ell)$ is the probability that
$\ell \in \{0,1\}^k$ is the transcript of the interaction between Alice and Bob under strategies
$f_A$ and $f_B$ and $\ip{\strat{A}}{\strat{B}}$ is the acceptance probability of the
interaction.
Conversely every vector $\strat{A} \in K_A$ corresponds to a probabilistic strategy $f_A$ for Alice
(and similarly for Bob, with $\ip{\strat{A}}{\strat{B}}$ being the acceptance probability of the  protocol).
\end{proposition}

\begin{proof}
The fact that $K_A$ and $K_B$ are {closed and }convex sets is straightforward from their definition.

We first show that strategies $f_A$ and $f_B$ can be converted into vectors in
$\tilde{K}_A$ and $\tilde{K}_B$ respectively. 
We define $\tildeptranscript{A}$ inductively as follows. Let
$\tildeptranscript{A}() = 1$.
Further let 
$\tildeptranscript{A}(i_1,\ldots,i_j,1) = 
\tildeptranscript{A}(i_1,\ldots,i_j,0) = 
\tildeptranscript{A}(i_1,\ldots,i_j)$ if $j$ is $j$ is odd.
Finally let 
$\tildeptranscript{A}(i_1,\ldots,i_j,1) = 
\tildeptranscript{A}(i_1,\ldots,i_j) \cdot f_A^{j}(i_1,\ldots,i_j)$ 
and 
$\tildeptranscript{A}(i_1,\ldots,i_j,1) = 
\tildeptranscript{A}(i_1,\ldots,i_j) \cdot (1- f_A^{j}(i_1,\ldots,i_j))$ if
$j$ is even.
Now let $\tildestrat{A}(i_1,\ldots,i_k) = 
\tildeptranscript{A}(i_1,\ldots,i_k)$.
It can be verified that $\tildestrat{A} \in \tilde{K}_A$ and indeed the vector
$\ptranscript{A} = 
\ptranscript{A,\tildestrat{A}}$
(as given by~\autoref{def:xa:xb}) is the vector $\tildeptranscript{A}$. Taking $\strat{A} = \tildestrat{A} * v$ gives us
the vector $\strat{A} \in K_A$ as needed.

Now for the more important direction, we claim that every vector $\strat{A} \in
K_A$ corresponds to a strategy $f_A$. Note by definition that since
$\strat{A} \in K_A$, there exists $\tildestrat{A} \in \tilde{K}_A$ such that
$\strat{A} = \tildestrat{A} * v$. Let $\tildeptranscript{A}$ be the vector
obtained from $\tildestrat{A}$ according to~\autoref{def:xa:xb}. For 
even $j$ and $i_1,\ldots,i_j \in \{0,1\}$. let
$f_A^{j}(i_1,\ldots,i_j) = 
\tildeptranscript{A}(i_1,\ldots,i_j,1)
/\tildeptranscript{A}(i_1,\ldots,i_j)$ for even $j$. (If 
$\tildeptranscript{A}(i_1,\ldots,i_j) = 0$, we define 
$f_A^{j}(i_1,\ldots,i_j) = 1$
Note that since 
$\tildeptranscript{A}(i_1,\ldots,i_t) \in [0,1]$ for every $t$ and
$\tildeptranscript{A}(i_1,\ldots,i_j) = 
\tildeptranscript{A}(i_1,\ldots,i_j,0) + 
\tildestrat{A}(i_1,\ldots,i_j,1)$, we have that 
$f_A^{j}(i_1,\ldots,i_j) \in [0,1]$ and so $f_A$ represents a strategy for Alice.
It can further be verified that if we apply the transformation of the previous
paragraph to this strategy $f_A$, we recover $\tildestrat{A}$ and so this
correspondence is indeed bidirectional.

Finally we verify that the acceptance probabilities are given by the inner
product function. Consider strategies given by vectors 
$\strat{A} \in K_A$ and $\strat{B} \in K_B$ with corresponding vectors 
$\tildestrat{A} \in \tilde{K}_A$ and $\tildestrat{B} \in \tilde{K}_B$ such that
$\strat{A} = \tildestrat{A} * v$
and $\strat{B} = \tildestrat{B} * v$.
Let $f_A$ and $f_B$ be the probabilistic strategies corresponding to $\strat{A}$
and $\strat{B}$ respectively as given by the previous paragraph. 
We claim, by induction on $j$, that for every $i_1,\ldots,i_j$ the probability
that the interaction reaches the partial transcript $i_1,\ldots,i_j$ under
strategies $f_A$, $f_B$ is $\tildeptranscript{A}(i_1,\ldots,i_j)\cdot
\tildeptranscript{B}(i_1,\ldots,i_j)$.
This is certainly true for $j=0$ where 
$\tildeptranscript{A}() = \tildeptranscript{B}() = 1$.
Now consider a partial transcript $i_1,\ldots,i_j$. 
By induction the probability of reaching this transcript is 
$p = \tildeptranscript{A}(i_1,\ldots,i_j) \cdot 
\tildeptranscript{B}(i_1,\ldots,i_j)$.
Suppose it is Alice's turn to speak. Then
$\tildeptranscript{B}(i_1,\ldots,i_j,0) = 
\tildeptranscript{B}(i_1,\ldots,i_j,1) = 
\tildeptranscript{B}(i_1,\ldots,i_j)$.
And 
$\tildeptranscript{A}(i_1,\ldots,i_j,0) +
\tildeptranscript{A}(i_1,\ldots,i_j,1) = 
\tildeptranscript{A}(i_1,\ldots,i_j)$.
Thus the probability of reaching the partial transcript
$i_1,\ldots,i_j,1$ is
$p \cdot f_A^{j}(i_1,\ldots,i_j) = 
\tildeptranscript{A}(i_1,\ldots,i_j) \cdot 
\tildeptranscript{B}(i_1,\ldots,i_j) \cdot 
\tildeptranscript{A}(i_1,\ldots,i_j,1)
/\tildeptranscript{A}(i_1,\ldots,i_j)
= 
\tildeptranscript{A}(i_1,\ldots,i_j,1) \cdot
\tildeptranscript{B}(i_1,\ldots,i_j)
= 
\tildeptranscript{A}(i_1,\ldots,i_j,1) \cdot
\tildeptranscript{B}(i_1,\ldots,i_j,1)$.
The calculation for the extension $i_1,\ldots,i_j,0$ is similar
and uses the fact that 
$\tildeptranscript{A}(i_1,\ldots,i_j) =
\tildeptranscript{A}(i_1,\ldots,i_j,0) +
\tildeptranscript{A}(i_1,\ldots,i_j,1)$, and
the probability that the transcript extends to
$i_1,\ldots,i_j,0$ conditioned on being at
$i_1,\ldots,i_j$ is $1 - f_A^j(i_1,\ldots,i_j)$.

From the above, we conclude that the probability of reaching a final transcript 
$\ell$ is 
$\tildeptranscript{A}(\ell)
\cdot \tildeptranscript{B}(\ell)
= 
\tildestrat{A}(\ell)\cdot
\tildestrat{B}(\ell)$. 
Thus the acceptance probability equals $\sum_{\ell: v(\ell) = 1}
\tildestrat{A}(\ell)\cdot
\tildestrat{B}(\ell)
= \sum_{\ell\in \{0,1\}^k} 
\tildestrat{A}(\ell)\cdot
\tildestrat{B}(\ell) \cdot v(\ell)^2 
= \sum_{\ell\in \{0,1\}^k} 
\strat{A}(\ell)\cdot \strat{B}(\ell)
= \langle 
\strat{A},\strat{B}\rangle$.
\end{proof}
\fi

\subsection{Upper bound on ISR in terms of PSR}\label{ssec:upperbound:isr:psr}

In this section we prove \autoref{thm:twoway-psr-isr-upper}.
Our first step is to prove that the \gapip problem (with
the right parameters) is hard for all problems with communication complexity
$k$. But first we define what it means for a promise problem to be
hard for some class of communication problems.

Recall that a promise problem $P = (P_n)_n$ is given by a collection of
$\yes$-instances $P_n^{\yes} \subseteq \{0,1\}^n \times \{0,1\}^n$ and
$\no$-instances $P_n^{\no} \subseteq \{0,1\}^n \times \{0,1\}^n$
with
$P_n^{\yes} \cap P_n^{\no} = \emptyset$.
We define below what it means for a promise problem $P$ to reduce to a
promise problem $Q$.

\begin{definition}
For promise problems $P = (P_n)_n$ and $Q = (Q_n)_n$ we say that $P$ reduces
to $Q$ if there exist functions $\ell\colon\N\to\N$ and
$f_n,g_n\colon\{0,1\}^n \to \{0,1\}^{\ell(n)}$ such that
if $(x,y) \in P_n^\yes$ then
$(f_n(x),g_n(y)) \in Q_{\ell(n)}^\yes$
and if $(x,y) \in P_n^\no$ then
$(f_n(x),g_n(y)) \in Q_{\ell(n)}^\no$.
We say $Q$ is hard for a class ${\cal C}$ if for every $P \in {\cal C}$
we have that $P$ reduces to $Q$.
\end{definition}

In other words Alice can apply $f_n$ to her input, and Bob can apply $g_n$ to
his input and get a new pair that is an instance of the $Q$-problem.
In particular if $Q$ has communication complexity $k$, then so does $P$.
This can be extended to functions $k(n)$ also: if $Q$ has communication
complexity $k(n)$, then $P$ has complexity $k(\ell(n))$.

Since we are mostly interested in $k$ being an absolute constant,
we do not strictly care about the length stretching function $\ell$. However, we note
that in the following proposition we only need a polynomial blowup (so
$\ell$ is a polynomial).

\begin{proposition}\label{prop:gapip:hard:psr}
For every positive integer $k$, $\gapip_{(2/3)2^{-k}, (1/3)2^{-k}}$ is
hard for $\capspsrcc(k)$.
\end{proposition}
\begin{proof}
Specifically we show that for any problem $P$ with inputs of length
$n$ and $\psrcc(P) \leq k$, there exist $N = \poly(n)$ and
transformations $f_n$ and $g_n$ such that $(x,y)$ is a $\yes$-instance of
$P$ if and only if $(f_n(x),g_n(y))$ is a $\yes$-instance of
$\gapip_{(2/3)2^{-k}, (1/3)2^{-k}}^N$.

Given $x \in \{0,1\}^n$ and random string $R$, let $X_R \in S_A^{(k)}$
describe the communication strategy of Alice with input $x$ and
randomness $R$. Similarly let $Y_R$ denote the strategy of Bob.
Recall that $\ip{X_R}{Y_R} = 1$ if the interaction accepts on randomness $R$
and
$\ip{X_R}{Y_R} = 0$ otherwise.
Let $f_n(x) = X$ be the concatenation of the strings $\{X_R\}_R$
and let $g_n(y) = Y$
be the concatenation of $\{Y_R\}_R$.
By Newman's Theorem we have that the number of random strings $R$ that we
need to consider is some polynomial $N' = \poly(n)$. Letting $N = 2^k \cdot
N^\prime$, we get that $X,Y \in \{0,1\}^N$ and $\ip{X}{Y} \geq (2/3)N' = (2/3)\cdot
2^{-k} \cdot N$ if $(x,y)$ is a \yes-instance of $P$
and
$\ip{X}{Y} \leq (1/3)N' = (1/3)\cdot
2^{-k} \cdot N$ if $(x,y)$ is a \no-instance of $P$. This gives the
desired reduction.
\end{proof}

Next we give an upper bound on $\isrcc(\gapip)$. In fact
we give an upper bound on
$\owisrcc(\gapip)$.

\begin{lemma}\label{lem:gaussian-upper}
For all $0 \leq s < c \leq 1$ and $\rho > 0$,
$\owisrcc[\rho]({\gapip^n_{c,s}}) = O(1 /\rho^2(c-s)^2)$.
\end{lemma}

\begin{proof}
Let $X \in \{0,1\}^n$ and {$Y\in\{0,1\}^n$} be the inputs to Alice and Bob. Recall that Alice
and Bob want to distinguish the case $\dotprod{X}{Y} \geq c\cdot n$ from the
case $\dotprod{X}{Y} \leq s \cdot n$.

We shall suppose without loss of generality that Alice and Bob have access to a source of $\rho$-correlated random spherical Gaussian vectors
$g, g^\prime \in \mathbb{R}^{n}$. We can enforce this in the limit by sampling several $\rho$-correlated random bit vectors $r_i,r'_i \in \{0,1\}^{n}$ for $i \in [N]$ and setting $g = \sum_{i =1}^N r_i/\sqrt{N}$ and $g^\prime= \sum_{i=1}^N r_i'/\sqrt{N}$. We leave out the details for this technical calculation (involving an appropriate use of the central limit theorem) here.

Let $t$ be a parameter to be chosen later and let $(g_1,g_1'), (g_2,g_2'),\ldots,(g_t,g_t')$ be $t$  independent $\rho$-correlated spherical Gaussian
vectors chosen from the source as above.
By the rotational invariance of the Gaussian distribution, we can assume without loss of generality that $g_i' = \rho g_i + \sqrt{1-\rho^2} g_i''$, where the $g_i''$'s are independent spherical Gaussian vectors.

As $g_1,\ldots,g_t$ are independent spherical Gaussians, by standard tail bounds (e.g., see Ledoux and Talagrand~\cite{LedouxT}),
with probability at least $1- 1/6$,
\[
\max_{i \in [t]} \dotprod{X}{g_i} = (\alpha \sqrt{\log t} \pm O(1))\cdot \sqrt{\dotprod{X}{X}}
\]
for some universal constant $\alpha$.

\noindent The protocol then proceeds as follows:
\begin{itemize}
  \item Alice computes $\ell = \arg\max_{i \in [t]} \dotprod{X}{g_i}$ and
$m$ such that $\dotprod{X}{X} \in ((m-1) \cdot \frac{(c-s)}{100}n, m\cdot \frac{(c-s)}{100}n]$ and sends $(\ell,m)$ to Bob {(note that this implies $m=\bigO{1/(c-s)}$)}.
  \item Bob accepts if $m \geq \frac{100c}{c-s}$ and $\dotprod{Y}{g^\prime_\ell} \geq \alpha \rho \sqrt{\log t} \cdot \frac{(c+s)n}{2\sqrt{m(c-s)(n/100)}}$ and rejects otherwise.
\end{itemize}
Now, write $Y = a X + b X^\perp$ for some vector $X^\perp$ with $a \dotprod{X}{X} = \dotprod{X}{Y}$ and $\dotprod{X}{X^\perp} = 0$. Then,
\[
  \dotprod{Y}{g_\ell'} = a \rho \dotprod{X}{g_\ell} + b \rho \dotprod{X^\perp}{g_\ell}+ \sqrt{1-\rho^2}\dotprod{Y}{g_\ell''}\;.
\]
As $\dotprod{X}{g_\ell}$ is independent of $\dotprod{X^\perp}{g_\ell}$ and $\dotprod{Y}{g_{\ell}''}$ {(the former from the fact that if $G$ is a spherical Gaussian and $u,v\in\R^n$ are orthogonal vectors then $\dotprod{u}{G}$ and $\dotprod{v}{G}$ are independent one-dimensional Gaussians)}, it follows from a simple tail bound for univariate Gaussians that with probability at least $1 - 1/6$, $\abs{\dotprod{X^\perp}{g_\ell}}$, $\abs{\dotprod{Y}{g_\ell''}} = O(\sqrt{n})$. By combining the above inequalities, we get that with probability at least $2/3$,
\[
\dotprod{Y}{g_\ell'} = \alpha \rho \sqrt{\log t} \dotprod{X}{Y}/\sqrt{\dotprod{X}{X}} \pm O(\sqrt{n}).
\]

To finish the proof observe that for \yes-instances, {$\dotprod{X}{{Y}} \geq cn$ (so that $m \geq \frac{100c}{c-s}$)} and $\dotprod{X}{Y}/\sqrt{\dotprod{X}{X}} \geq \beta_1 \triangleq c\cdot n/\sqrt{m (c-s) (n/100)}$; while for
\no-instances, $\dotprod{X}{Y}{/\sqrt{\dotprod{X}{X}}} \leq \beta_2 \triangleq s\cdot n/\sqrt{(m-1)(c-s)(n/100)}$. Hence, the protocol works correctly if $\alpha\rho\sqrt{\log t}(\beta_1 - \beta_2) \gg O(\sqrt{n})$.

It follows from the settings of parameters that this indeed happens for some
$\log t = \Theta(1/(\rho^2 (c-s)^2))$. In particular, we have
\[
\beta_1 - \beta_2 = \frac{cn - sn}{\sqrt{m(c-s)(n/100)}} -
\frac{sn}{\sqrt{(c-s)(n/100)}} \left(\frac{1}{\sqrt{m-1}} - \frac{1}{\sqrt{m}}\right).
\]
By the condition $m \geq \frac{100c}{c-s}$ we have $\frac{1}{\sqrt{m-1}} - \frac{1}{\sqrt{m}} \leq \frac{c-s}{2s}$ and thus
\[
\beta_1 - \beta_2 \geq \frac12 \frac{cn - sn}{\sqrt{m(c-s)(n/100)}}
= \frac{\sqrt{25(c-s)n}}{\sqrt{m}}.
\]
And so when {$\log t \gg \Omega(1/(\alpha^2 \rho^2 (c-s)^2))$} we find
$\alpha\rho\sqrt{\log t}(\beta_1 - \beta_2) \gg O(\sqrt{n})$ as required.
\end{proof}

The above lemma along with the hardness of $\gapip$ gives us \autoref{thm:twoway-psr-isr-upper}.
\begin{proof}[Proof of \autoref{thm:twoway-psr-isr-upper}]
By \autoref{prop:gapip:hard:psr}, for every promise problem $P$ such that $\psrcc(P) \leq k$, $P$ reduces to
$\gapip_{c,s}$ with $c = (2/3)2^{-k}$ and $s = (1/3)2^{-k}$.
By \autoref{lem:gaussian-upper} we get that the reduced instance of
$\gapip_{c,s}$ has a one-way isr communication protocol of
with $O_\rho(1/(c-s)^2) = O_\rho({2^{2k}})$ bits of communication.
The theorem follows.
\end{proof}

\subsection{ISR lower bound for \sparsegapip}\label{ssec:lowerbound:isr}

In this section, we consider the promise problem
$\sparsegapip_{.99q,.9q^{-1},.6q^{-1}}^n$ and show that it has a
one-way psr protocol with $O(\log q)$ bits of communication, and then give a
two-way isr lower bound of $q^{\Omega(1)}$ for this problem. Together this
proves \autoref{thm:psr-isr-lower}.

\begin{proposition}\label{prop:sparsegapinnerpsr}
$\forall c > s$ and $\forall q,n$,
we have
\[
\owpsrcc(\sparsegapip_{q,c,s}^n) \leq \bigO{\frac{1}{q^2(c-s)^2}\left(\log \frac{1}{c} + \log \frac{1}{q(c-s)} + \log\log \frac{c}{c-s}\right)}\;.
\]
\end{proposition}
\begin{proof}[Proof (Sketch)]
We first show that there exists an atomic one-way communication protocol for the problem
$\sparsegapip^n_{q,c,s}$ with the following features (where $\gamma=\bigTheta{(c-s)/c}$):
\begin{enumerate}
\item the length of communication is $\bigO{\log 1/c + \log 1/(q(c-s)) + \log\log 1/\gamma}$.
\item $\yes$-instances are accepted with probability at least $(1-\gamma)\cdot \frac{c}{c-s}\cdot {\frac{100}{m}}$ and $\no$-instances with probability at most $\frac{s}{c-s}\cdot\frac{100}{m-1}$ for some $m = \bigOmega{ c/(c-s) }$ known by both parties. In particular, the difference between completeness and soundness is $\bigOmega{1/m}$.
\end{enumerate}
The atomic protocol lets the shared randomness determine a sequence of $t\eqdef -{\log(1/\gamma)}/{\log(1-c)}$ indices 
$i_1,i_2,\ldots,i_t$ in $[n]$.  Alice first computes $m=\bigO{1/(c-s)}$ such that $\normtwo{x}^2 \in ((m-1) \cdot \frac{(c-s)}{100}n, m\cdot \frac{(c-s)}{100}n]$, and picks the smallest index $\ell$ such that $x_{i_\ell} \neq 0$. Then she sends $(\ell,m)$ to Bob, or $(0,0)$ if no such index was found. (Note that by sparsity of $x$, we have $m=\bigO{1/(q(c-s))}$).
Bob outputs 0 if he received $(0,0)$ or if $m < \frac{100c}{c-s}$, and {the value of} $y_{i_\ell}$ otherwise.

The completeness follows from the fact that, for \yes-instances, $\normtwo{x}^2 \geq cn$ (implying $m \geq \frac{100c}{c-s}$) and one expects an index
$\ell$ such that $x_{i_\ell} \ne 0$ among the first roughly $1/c$ choices of $\ell$, so that this will be the case with high probability among the first $t$; conditioned on this {(which happens with probability at least $1-\gamma$ by our choice of $t$, as the probability no $\ell$ is found is upper bounded by $(1-\normtwo{x}/n)^t\leq (1-c)^t$)},
$y_{i_\ell}$ is $1$ with probability at least $\frac{cn}{\normtwo{x}^2} \geq \frac{c}{c-s}\frac{100}{m}$. As for the soundness, observe that a \no-instance for which Alice does not send $0$ to Bob will have $y_{i_\ell}=1$ with probability at most $\frac{sn}{\normtwo{x}^2} < \frac{s}{c-s}\cdot\frac{100}{m-1}$. Now, since $m \geq \frac{100c}{c-s}$, $\frac{100s}{c-s}\left(\frac{1}{m-1} - \frac{1}{m}\right) \leq \frac{100}{3m}$; and by {the} choice of {$\gamma \leq \frac{c-s}{3c}$} we also have ${\gamma}\frac{c}{c-s}\frac{100}{m} \leq \frac{100}{3m}$. This implies the difference in acceptance probability between completeness and soundness is at least $\frac{100}{3m}$.

\noindent Repeating this protocol $O(m^2)=O(1/(q^2(c-s)^2))$ times and thresholding yields the final result.
\end{proof}

\noindent We now state our main lower bound theorem.
\begin{theorem}\label{thm:main-neg}
There exists $\eps > 0$ such that {$\forall \rho\in [0,1)$ and} $\forall q$, {there exists $N$ for which the following holds. For every $n\geq N$,} we have
\[
  \isrcc_{\rho}(\sparsegapip_{.99q,.9q^{-1},.6q^{-1}}^n)  \geq \eps\cdot\sqrt{q}.
\]
\end{theorem}
We prove \autoref{thm:main-neg} in \autoref{sec:proof:theo:main-neg}, but we now note that
\autoref{thm:psr-isr-lower} follows immediately from \autoref{prop:sparsegapinnerpsr} and
\autoref{thm:main-neg}.

\begin{proof}[Proof of \autoref{thm:psr-isr-lower}]
The promise problem is $P = \sparsegapip_{.99\cdot2^k,.9\cdot 2^{-k},.6\cdot 2^{-k}}$.
By \autoref{prop:sparsegapinnerpsr} we have $\owpsrcc(P) \leq \bigO{k}$
and by \autoref{thm:main-neg} we have $\isrcc(P) \geq 2^{\bigOmega{k}}$.
\end{proof}

\section{Proof of \autoref{thm:main-neg}}\label{sec:proof:theo:main-neg}

Our goal for this section is to prove, modulo some technical theorems,
that \sparsegapip has high communication complexity in the
{imperfectly} shared randomness setting. Before jumping into the proof
we give some overview first.

\subsection{Proof setup}
To prove \autoref{thm:main-neg}, we will show that for every
``strategy'' of Alice and Bob, there is a pair of distributions $\dyes$
and $\dno$ supported (mostly) on \yes\ and \no\ instances, respectively, such that the strategies
do not have much ``success'' in distinguishing them. We note that in contrast to typical lower bounds for perfectly shared randomness, we cannot hope to fix a distribution that works against every strategy. Indeed for every pair of distributions, by virtue of the protocol given in
\autoref{prop:sparsegapinnerpsr}
and the Yao min-max principle we have even a deterministic strategy
(let alone randomized strategy with imperfect sharing) that succeeds in distinguishing them with high probability. So instead we have to fix the strategies first and then give a pair of distributions that does not work for that strategy.  We define the notion of
strategy and success more formally below, and then work towards the proof of
\autoref{thm:main-neg}.

\paragraph{Strategy:} We now use \autoref{ssec:strategy:preliminaries} to formalize
what it would mean to have a
$k$-bit communication protocol for any communication problem. For aesthetic reasons we view Alice and Bob's strategies as probabilistic
ones.
Recall, by \autoref{prop:ka:kb},
that $k$-bit probabilistic communication strategies for Alice can be described by
elements of $K_A^{(k)} \subseteq [0,1]^{2^k}$ and similarly
by elements of $K_B^{(k)} \subseteq [0,1]^{2^k}$ for Bob. 
So, on randomness $r$ we have that Alice's
communication strategy can be described by a function $f^{(r)}\colon\{0,1\}^n
\to K_A^{(k)}$. Similarly for randomness $s$, Bob's communication strategy
can be described by a function $g^{(s)}\colon\{0,1\}^n \to K_B^{(k)}$.

Thus,  a \emph{strategy} for a game is a pair of sets of functions $\mathcal{F}=(f^{(r)})_r,\mathcal{G}=(g^{(s)})_s$, where
\begin{align*}
	f^{(r)} &\colon \{0,1\}^n\to K_A^{(k)} \\
	\mbox{ and } g^{(s)} &\colon \{0,1\}^n\to K_B^{(k)}.
\end{align*}

We consider a pair of distributions $D=(\dyes, \dno)$ to be {\em valid} if
$\dyes$ is mostly (say with probability $.9$) supported
on \yes-instances and $\dno$ mostly on \no-instances.
For valid $D$, we define {the \emph{success} as}
\begin{align*}
\success_{D}(f,g) &\eqdef
			\shortexpect_{(x,y)\sim \dyes}\left[\dotprod{ f(x) }{ g(y) }\right]  -
			\shortexpect_{(x,y)\sim \dno }\left[\dotprod{ f(x) }{ g(y) }\right]\\
\success_{D,\rho}(\mathcal{F},\mathcal{G}) &\eqdef \abs{
			\shortexpect_{r\sim_\rho s}
			\left[\success_{D}(f^{(r)},g^{(s)})\right]
		 } \\
\success_{\rho}(\mathcal{F},\mathcal{G}) &\eqdef \min_{{\rm valid~}D}\success_{D,\rho}(\mathcal{F},\mathcal{G})
\end{align*}

We note that any strategy that distinguishes \yes-instances of
\sparsegapip from \no-instances with probability
$\eps$ must have success $\eps - .1$ on every valid distribution as
well (with the difference of $.1$ coming up due to the fact that valid
distributions are not entirely supported on the right instances).
In what follows we will explain why strategies (with small $k$) do not have
sufficiently positive success.

\subsection{Overview of proof of \autoref{thm:main-neg}.}
To prove \autoref{thm:main-neg} we need to show that if a pair of
strategies $(\calf,\calg)$ achieves $\success_{\rho}(\calf,\calg) >
.01$ then $k$ must be large.
Roughly our strategy for showing this is as follows: We first
define two simple distributions $\dyes$ and $\dno$ (independent of the
strategy $(\mathcal{F},\mathcal{G})$) and show that any fixed pair of
functions $(f,g)$ that are successful in
distinguishing $\dyes$ from $\dno$ must have a few influential variables
and furthermore at least one of these variables must be common to both $f$ and $g$
(see \autoref{thm:singlefg}). Our proof of this theorem, is based on the ``invariance
principle'' \cite{MO:GAFA:10} and
\autoref{thm:singlefg} is a variant of it which is {particularly} suited for use in
communication complexity. The proof of this theorem is deferred to \autoref{sec:invariance}.

We use this theorem to design agreement distillation strategies for two new players Charlie and Dana as follows: Given shared random pair
$(r,s)$, Charlie picks a random influential variable $x_i$ of the function
$f^{(r)}$ used by Alice on random string $r$ and outputs the index $i \in [n]$.
Dana similarly picks a random influential variable $y_j$ of the function
$g^{(s)}$ used by Bob and outputs $j$. \autoref{thm:singlefg} assures us that with non-trivial
probability $i = j$ and this gives an agreement protocol.

If we could argue
that $i = j$ has high min-entropy, then we would be done (using
\autoref{lem:lb:agreement} which asserts that it is not possible to distill agreement with high-entropy
and high probability).
But this step is not immediate (and should not be since we
have not crafted a distribution specific to $(\calf, \calg)$). To show that
this strategy produces indices of high min-entropy, we consider the
distribution of indices that is produced by Charlie as we vary $r$ and
let $\bad_C$ denote the indices that are produced with too high a
probability. Similarly we let $\bad_D$ denote the indices that are produced
with too high a probability by Dana. We now consider a new distribution
$\dyes'$ supported on \yes-instances of the
{$\sparsegapip$} problem.
In $\dyes'$ the $(x,y)$ pairs are chosen so that when restricted to coordinates in
$\bad_C \cup \bad_D$ they look like they come from $\dno$
while when restricted to coordinates outside
$\bad_C \cup \bad_D$ they look like they come from $\dyes$
(see \autoref{defn:dyesprime} below for a precise description).
Since $\bad_C \cup \bad_D$ is small, the distribution $\dyes'$ remains
supported mostly on \yes-instances, but strategies that depend mainly
on coordinates from $\bad_C \cup \bad_D$ would not have much success in
distinguishing $\dyes'$ from $\dno'$ (which remains the original $\dno$).

We use this intuition to argue
formally in \autoref{lem:reduction} that a slightly modified sampling protocol of Charlie and Dana,
where they discard
$i,j$ from $\bad_C \cup \bad_D$,
leads to agreement with noticeably high probability
on a high-entropy random variable, yielding the desired contradiction.

In the rest of this section we first present the main definitions needed to
state \autoref{thm:singlefg}.
We then prove
\autoref{thm:main-neg} assuming \autoref{thm:singlefg}. We prove {the latter} in \autoref{sec:invariance}, {along with the main technical ingredient it relies on, the invariance principle of \autoref{thm:invariance:principle:bifunction}}.

\subsection{Background on influence of variables}

We now turn to defining the notion of influential variables for
functions and related background material for functions defined on
product probability spaces.

Recall that a finite probability space is given by a pair $(\Omega,\mu)$ where $\Omega$ is a finite
set and $\mu$ is a probability measure on $\Omega$.
We will begin with the natural
probabilistic definition of influence of a variable on functions defined on product spaces,
and then relate it to a more algebraic
definition which is needed for the notion of low-degree influence.

\begin{definition}[Influence and variance]
Let $(\Omega,\mu)$ be a finite probability space, and let $h\colon\Omega^n \to \R$ be a function on product probability space.
The \emph{variance} of $h$, denoted $\Var(h)$, is defined as the variance of the random variable $h(x)$ for $x \in \Omega^n \sim \mu^{\otimes n}$, i.e., $\Var(h) = \E_{x} [ h(x)^2] - \left(\E_{x} [ h(x) ]\right)^2$.

\noindent For $i \in [n]$, the \emph{$i$-th influence of $h$} is defined as
\[ \infl[i]{h} = \E_{x^{(-i)} \sim \mu^{\otimes (n-1)}}\Bigl[ \Var_{x_i \sim \mu} [ h(x) ] \Bigr] \]
where $x^{(-i)}$ denotes all coordinates of $x$ except the $i$'th coordinate.
\end{definition}

To define the notion of {\em low-degree} influence, we need to
work with a multilinear representation of functions $h\colon \Omega^n \to \R$.
Let $b = |\Omega|$ and
$\mathcal{B} = \{\chi_0,\chi_1,\dots,\chi_{b-1}\}$ be a basis of real-valued functions over $\Omega$. Then, every function $h\colon\Omega^n \to \R$ has a unique multilinear expansion of the form
  \begin{equation}\label{eq:mult-expansion}
   h(x) = \sum_{\multindex = (\sigma_1,\dots,\sigma_n) \in \{0,1,\dots,b-1\}^n} \hat{h}_{\multindex} \chi_{\multindex}(x)
  \end{equation}
 for some real coefficients $\hat{h}_{\multindex}$, where
$\chi_{\multindex}$ is given by $\chi_{\multindex}(x) \eqdef \prod_{i \in [n]} \chi_{\sigma_i}(x_i)$.

 When the ensemble $\mathcal{B}$ is a collection of {\em orthonormal} random variables, namely $\chi_0=1$ and $\E_{a \sim \mu} [\chi_{j_1}(a) \chi_{j_2}(a)] = \delta_{j_1 j_2}$, it is easy to check that
  $\Var(h) = \sum_{\multindex \neq \mathbf{0}} \hat{h}^2_{\multindex}$ and also that
\[ \infl[i]{h} =  \sum_{\multindex:\,\sigma_i \neq 0} \hat{h}^2_{\multindex} \ . \]
One can also take the above as the algebraic definition of influence, noting that it is
independent of the choice of the orthonormal basis $\mathcal{B}$ and thus well-defined.
The degree of a multi-index $\multindex$ is defined as $|\multindex| = \abs{ \setOfSuchThat{ i }{ \sigma_i \neq 0 } }$, and this leads to the definition of low-degree influence.

\begin{definition}[Low-degree influence]
For a function $h\colon\Omega^n \to \R$ with multilinear expansion as in \eqref{eq:mult-expansion} with respect to any orthonormal basis, the \emph{$i$-th degree $d$ influence of $h$} is the influence of the truncated multilinear expansion of $h$ at degree $d$, that is
\[
\infldeg[i]{d}{h} \eqdef \sum_{\multindex:\,\sigma_i \neq 0 \atop {|\multindex| \leq d}} \hat{h}_{\multindex}^2 \ .
\]
\end{definition}

\begin{remark}[Functions over size $2$ domain]
When $|\Omega|=2$, and $\{1,\chi\}$ is an orthonormal basis of real-valued functions over $\Omega$, the expansion \eqref{eq:mult-expansion} becomes the familiar Fourier expansion
$h(x) = \sum_{S \subseteq [n]} \widehat{h}(S) \prod_{i \in S} \chi(x_i)$, and we have $\infl[i]{h} \eqdef \sum_{S  \ni i} \widehat{h}(S)^2$
and $\infldeg[i]{d}{h} \eqdef \sum_{\substack{S \ni i\\\abs{S}\leq d}} \widehat{h}(S)^2$.
\end{remark}

\noindent We will make use of the following simple upper bound on the number of low-degree influential coordinates (which follows immediately, for instance, from \cite[Proposition 3.8]{MO:GAFA:10}) 
\begin{proposition}
\label{prop:low-inf}
For every $\tau > 0$ and $d \in \N$  there exists $t = t(\tau,d)$ such that
for all $n$ and all functions
$h\colon \Omega^n \to [-1,1]$, we have $\abs{ \setOfSuchThat{i \in [n] }{ \infldeg[i]{d}{h} > \tau } } \leq
t$. {(Furthermore, one can take $t = d/\tau$).}\end{proposition}

For the invariance principle, we will understand the behavior of a function when its domain is replaced by a different probability space with matching second moments. For this purpose, we will view functions as multilinear polynomials as follows.

\begin{definition}[Functions on product spaces as multilinear polynomials]\label{def:rv-to-poly}
The multilinear polynomial associated with a function $h\colon\Omega^n \to \R$ with respect to a basis $\mathcal{B} = \{\chi_0,\chi_1,\dots,\chi_{b-1}\}$ of real-valued functions over $\Omega$ is a polynomial in indeterminates $\mathbf{z} = \setOfSuchThat{ z_{i,j} }{ i \in [n], j \in \{0,1,\dots,b-1\} }$ given by
\[ H(\mathbf{z}) = \sum_{\multindex \in \{0,1,\dots, b-1\}^n} \hat{h}_{\multindex}\mathbf{z}_{\multindex} , \]
$\mathbf{z}_{\multindex}$ stands for the monomial $\prod_{i=1}^n z_{i,\sigma_i}$ and the coefficients $\hat{h}_{\multindex}$ are given by the multilinear expansion \eqref{eq:mult-expansion} of $f$ w.r.t. $\mathcal{B}$.
\end{definition}

Above, we saw how a function can be viewed as a multilinear polynomial w.r.t. a basis of random variables. Conversely, one can view multilinear polynomials as functions by substituting random variables for its indeterminates.
\begin{definition}[Multilinear polynomials as random variables on product spaces]\label{def:poly-to-rv}
Given a collection of random variables $\mathcal{X} = \{\chi_0,\dots,\chi_{{b-1}}\}$ over a probability space $(\Omega,\mu)$, one can view a multilinear polynomial $P$ in indeterminates $\mathbf{z} = \setOfSuchThat{ z_{i,j} }{ i \in [n], j \in \{{0,1,\dots,b-1}\} }$ given by
\[ P(\mathbf{z}) = \sum_{\multindex \in \{{0,1,\dots, b-1}\}^n} \hat{P}_{\multindex}\mathbf{z}_{\multindex} , \]
where $\mathbf{z}_{\multindex}$ stands for the monomial $\prod_{i=1}^n z_{i,\sigma_i}$,
 as a random variable $P(\mathcal{X}^n)$ over the probability space $(\Omega^n,\mu^{\otimes n})$ mapping $x =(x_1,\dots,x_n)$ to
  \begin{equation}\label{eq:poly-to-fn}
    \sum_{\multindex \in \{{0,1,\dots, b-1}\}^n} \hat{P}_{\multindex} \prod_{i=1}^n \chi_{\sigma_i}(x_i) \ .
  \end{equation}
\end{definition}

\subsection{Proof of \autoref{thm:main-neg}}

We start by introducing a few definitions, in particular of the central distributions and
the extraction strategy. We {begin with the description of} the basic distributions
$\dyes$ and $\dno$.

\begin{definition}[$\dyes$, $\dno$]
\label{defn:dno}
We define two distributions $B_N$ and $B_Y$ on
$\{0,1\}\times\{0,1\}$ below.
The distributions $\dyes$ and $\dno$ will be product distributions on
$\left(\{0,1\}\times\{0,1\}\right)^n$, given by
$\dyes = B_Y^{\otimes n}$
and
$\dno = B_N^{\otimes n}$.
\begin{itemize}
	\item A pair $(x,y)$ is drawn from $B_N$ by setting $x\sim\bernoulli{1/q}$
\\		and $y\in\{0,1\}$ uniformly at random. Note that $x,y$ are independent, and $\expect{xy} =\frac{1}{2q} $.
	\item A pair $(x,y)$ is drawn from $B_Y$ by setting
		\begin{align*}
			(x,y) &= \begin{cases}
				(0,1) &\text{ w.p. } \frac{1}{2}\left(1-\frac{1.95}{q}\right)\\
				(0,0) &\text{ w.p. } \frac{1}{2}\left(1-\frac{0.05}{q}\right)\\
				(1,1) &\text{ w.p. } \frac{1.95}{2q}\\
				(1,0) &\text{ w.p. }\frac{.05}{2q}
			\end{cases}
		\end{align*}
		so that the marginals of $x$, $y$ in $B_Y$ match those of $B_N$, and $\expect{xy} = \frac{1.95}{2q}$.
\end{itemize}
\end{definition}

A straightforward application of tail inequalities for independent, identically distributed (i.i.d.) random variables
tells us that $\dyes$ is mostly supported on \yes-instances of
{$\sparsegapip^n_{.99q,0.9q,0.6q}$} with high probability
for sufficiently large $n$.
Similarly $\dno$ is mostly supported on $\no$-instances.

Our main technical result is the following theorem showing any fixed pair of vector-valued functions $(f,g)$ (corresponding to strategies for Alice and Bob) that succeed in distinguishing $\dyes$ from $\dno$ must share an influential
variable (with non-trivially high influence of non-trivially low-degree).

\begin{theorem}\label{thm:singlefg}
There exist functions {$k_0 \geq \Omega_\eps(\sqrt{q})$},  $d(q,\eps) < \infty$,
and $\tau(q,\eps) > 0$,
defined for $q \in \N$,and $\eps > 0$,
such that the following holds:
For every $\eps > 0$ and $k, q \in \N$ and every sufficiently
large $n$,
if $k < k_0(q,\eps)$ and
$f\colon \{0,1\}^n\to K_A^{(k)}$ and $g\colon \{0,1\}^n \to K_B^{(k)}$ are
functions such that $\success_{(\dyes,\dno)}(f,g) \geq \eps$,
then there exists $i \in [n]$ such that
\[
\min\left\{
  \max_{j \in [{2^k}]} \infldeg[i]{d(q,\eps)}{f_j},
  \max_{j \in [{2^k}]} \infldeg[i]{d(q,\eps)}{g_j}
\right\}  \geq \tau(q,\eps)
\]
where $f_j$ and $g_j$ denote the $j$'th component function of $f$ and $g$, respectively.
(Here, the influence of $f_j$ is w.r.t. to the $\bernoulli{1/q}$ distribution on $\{0,1\}$, and that of $g_j$ is w.r.t. the uniform distribution on $\{0,1\}$.)
\end{theorem}

This theorem is proved in \autoref{sec:invariance}. Building on this theorem, we can
try to build agreement distillation protocols $(\ext_C,\ext_D)$ that exploit
the success of the strategies $(\calf,\calg)$ to distill common randomness.
We start by first identifying coordinates that may be influential for too many
pairs $(r,s)$ (and thus may be produced with too high a probability by a naive distillation
protocol).

For the rest of the section we fix $q \in \N$ and $\eps > 0$ and let $d = d(q,\eps)$
and $\tau = \tau(q,\eps)$ where $d(\cdot,\cdot)$ and $\tau(\cdot,\cdot)$ are the functions
from \autoref{thm:singlefg}.

\begin{definition}[$\bad_C$, $\bad_D$]
Let $\delta = 1/(100\cdot 2^{k_0} t)$ where $t = t(\tau,d)$ as given by~\autoref{prop:low-inf}, and $k_0 = k_0(q,\eps)$ is given by \autoref{thm:singlefg}. Define
\[
  \textsc{Bad}_C \eqdef \setOfSuchThat{ i\in[n] }{ \probaDistrOf{r}{ \max_{j \in [2^k]} \infldeg[i]{d}{{f^{(r)}_j}} > \tau } > \frac{1}{\delta n}
  }
  \qquad
\mbox{ and } \]
\[
  \textsc{Bad}_D \eqdef \setOfSuchThat{ i\in[n] }{ \probaDistrOf{s}{ \max_{j \in [2^k]} \infldeg[i]{d}{g^{(s)}_j} > \tau } > \frac{1}{\delta n} } ,
\]
where $r,s$ denote the randomness available to Alice and Bob, $f^{(r)}_j$ denotes the $j$'th component function for Alice's strategy on randomness $r$, and similarly for $g^{(s)}_j$.
\end{definition}

\noindent Directly from this definition {and~\autoref{prop:low-inf}}, we get
\begin{proposition} \label{prop:t-of-tau-d}
$|\bad_C|,|\bad_D| \leq  {2^k} \cdot t\cdot\delta\cdot n \le  n/100$.
\end{proposition}

\noindent Next, we define the extraction distillation protocols for Charlie and Dana:
\begin{definition}[$(\mathsf{Ext}_C,\mathsf{Ext}_D)$]
\label{defn:ext}
For $r \in \{0,1\}^*$, let
\[
  S_r \eqdef \setOfSuchThat{ i\in [n]\setminus\textsc{Bad}_C }{ \max_{j \in [{2^k}]} \infldeg[i]{d}{f^{(r)}_j} > \tau }
  ~~\text{and}~~
  T_s \eqdef \setOfSuchThat{ i\in [n]\setminus\textsc{Bad}_D }{\max_{j \in [{2^k}]} \infldeg[i]{d}{g^{(s)}_j} > \tau }\;.
\]
Then, $\mathsf{Ext}_C(r)$ is defined as follows:
\begin{quote}
if $S_r=\emptyset$ output $i\sim\uniform_{[n]}$; otherwise output $i\sim\uniform_{S_r}$.
\end{quote}
$\mathsf{Ext}_D(s)$ is defined similarly:
\begin{quote} if $T_{s}=\emptyset$ output $j\sim\uniform_{[n]}$; otherwise output $j\sim\uniform_{T_{s}}$.
\end{quote}
\end{definition}

\begin{proposition}
\label{prop:entropy}
$H_\infty(\mathsf{Ext}_C(r)) \geq \log n - \log (1 + 1/\delta)$.
\end{proposition}
\begin{proof}
Fix $i \in [n] \setminus (\bad_C \cup \bad_D)$. We have
\[\probaOf{ i\mbox{ is output} } \leq \probaOf{i \in S_r\mbox{ and }i\mbox{ is output} } +
\probaCond{ i \mbox{ is output} }{ S_r = \emptyset } \leq 1/(\delta n) + 1/n\]
{where the upper bound on the first term comes from observing that as $i\notin \bad_C$, $\probaOf{i \in S_r} \leq 1/(\delta n)$.}
The proposition follows.
\end{proof}

Finally we turn to proving that $\ext_C$ and $\ext_D$ do agree with
non-trivial probability.
To do so we need to
consider a new distribution on \yes-instances, defined next:

\begin{definition}[$\dyes'$]
\label{defn:dyesprime}
The distribution $\dyes'$ is a product distribution on
$\left(\{0,1\}\times\{0,1\}\right)^n$, where $(x_i,y_i) \sim
B_N$ if $i \in \bad_C \cup \bad_D$ and $(x_i,y_i) \sim B_Y$ otherwise.
\end{definition}

Using \autoref{prop:t-of-tau-d} above we have that $\E_{i,x,y}[x_iy_i]\geq .93/q$ {(where the expectation is over $(x,y)\sim\dyes'$ and $i$ drawn uniformly at random from $[n]$)} and so by standard tail inequalities we still
have that $\dyes'$ is mostly supported on \yes-instances. Our main lemma for this section is that if $(\calf,\calg)$ are successful in distinguishing $\dyes'$ and $\dno$ and $k$ is small, then $\ext_C$ and $\ext_D$ are likely to agree with noticeable probability (which would contradict
\autoref{lem:lb:agreement}).

\begin{lemma}
\label{lem:reduction}
Let $k_0 = k_0(q,\eps)$, $d = d(q,\eps)$ and $\tau = \tau(q,\eps)$ be
as given in \autoref{thm:singlefg}, and let $t = t(\tau,d)$ {as given by~\autoref{prop:low-inf}}.
If $\success_{(\dyes',\dno),\rho}(\calf,\calg) \geq 2\eps$,
and $k < k_0$
then
\[ \Pr_{{r \sim_\rho s}} [\ext_C(r) = \ext_D(s)] \geq \eps/(2^{2k}t^2) \ . \]
\end{lemma}

\begin{proof}
Expanding the definition of $\success(\cdot,\cdot)$, we have
\[
\left\lvert \E_{(r,s)}\left[
  \E_{(x,y) \sim \dyes'} \left[ \dotprod{ f^{(r)}(x) }{ g^{(s)}(y) }\right]
  - \E_{(x,y) \sim \dno} \left[ \dotprod{ {f^{(r)}}(x) }{ g^{(s)}(y) }\right]
\right]
\right\rvert
\geq 2\eps.
\]
Say that a pair $(r,s)$ is \good if
\[
\left\lvert \E_{(x,y) \sim \dyes'} \left[ \dotprod{ f^{(r)}(x) }{ g^{(s)}(y) }\right]
- \E_{(x,y) \sim \dno} \left[ \dotprod{ f^{(r)}(x) }{ g^{(s)}(y) }\right]
\right\rvert
\geq \eps.
\]
By a Markov argument we thus have
\[
\probaDistrOf{(r,s)}{ (r,s) \textrm{ is \good} } \geq \eps.
\]
For any fixed \good $(r,s)$ we now prove that there exists $i \in
(S_r \cap T_{s}) \setminus (\bad_C \cup \bad_D)$. Note that once we have such
an $i$, we have that
$\Pr[\ext_C(r) = \ext_D(s) = i]$ with probability at least $1/(2^{k}t(\tau,d))^2$ by~\autoref{prop:low-inf}. 
Combining this with the probability that $(r,s)$ is good, we have
$\Pr_{(r,s)}[\ext_C(r) = \ext_D(s)] \geq \eps/(2^{2k}t(\tau,d)^2)$ which
yields the lemma. So we turn to this claim.

To simplify notation, assume without loss of generality that
$\bad_C \cup \bad_D = \{m+1,\ldots,n\}$. Define functions $f_1\colon \{0,1\}^m
\to K_A^{(k)}$ and $g_1\colon \{0,1\}^m \to K_B^{(k)}$ by letting
\[
  f_1(x) = \E_{z \sim \bern{n-m}{1/q}}[f^{(r)}(x\cdot z)]
    \quad \text{and} \quad
  g_1(y) = \E_{w \sim \mathcal{U}(\{0,1\}^{n-m})}[g^{(s)}({y}\cdot w)]
\]
{where $u\cdot v$ denotes the concatenation of $u$ and $v$.}
Note that the success of $(f^{(r)},g^{(s)})$ in distinguishing 
$\dyes'$ from $\dno$ turns into the success of $(f_1,g_1)$ in distinguishing
$\dyes_m$ from $\dno_m$ (where $\dyes_m = B_Y^{\otimes m}$
and $\dno_m = B_N^{\otimes m}$) --- this is immediate since
$(x\cdot z,y \cdot w) \sim \dyes'$ if $(x,y) \sim \dyes_m$ and
$(x\cdot z,y \cdot w) \sim \dno$ if $(x,y) \sim \dno_m$.

So we have $\success_{(\dyes_m,\dno_m)}(f_1,g_1) \geq \eps$.
Since $k < k_0$ we have that there must exist a variable $i \in [m]$ and indices $j,j' \in [{2^k}]$
with $\infldeg[i]{d}{f_{1,j}} > \tau$
and
$\infldeg[i]{d}{g_{1,j^\prime}} > \tau$. 
(Here $f_{1,j}$ is the $j$'th component function of $f_1$, and similarly for $g_{1,j'}$.) Indeed, this follows from~\autoref{thm:singlefg}, and~\autoref{prop:t-of-tau-d} (which ensures that $m\geq \frac{98}{100}n$, and therefore sufficiently large for the conclusion of the theorem to hold).
But
$\infldeg[i]{d}{f^{(r)}_j} \geq \infldeg[i]{d}{f_{1,j}}$
and
$\infldeg[i]{d}{g^{(r^\prime)}_{j'}} \geq \infldeg[i]{d}{g_{1,j'}}$.
To see this, note that $\widehat{f_{1,j}}(S) = \widehat{f^{(r)}_j}(S)$ for $S \subseteq [m]$
and so
\begin{align*}
\infldeg[i]{d}{f^{(r)}_j}
& =  \sum_{i \in S \subseteq [n], |S| \leq d} \widehat{f^{(r)}_j}(S)^2 \\
& \geq  \sum_{i \in S \subseteq [m], |S| \leq d} \widehat{f^{(r)}_j}(S)^2 \\
& =  \sum_{i \in S \subseteq [m], |S| \leq d} \widehat{f_{1,j}}(S)^2 \\
& =  \infldeg[i]{d}{f_{1,j}}.
\end{align*}
We thus conclude that $i \in S_r \cap T_{s} \cap [m]$ and this concludes
the claim, and thus the lemma.
\end{proof}

\begin{proof}[Proof of~\autoref{thm:main-neg}]
The proof follows easily from \autoref{lem:lb:agreement}
and \autoref{lem:reduction}.
Assume for contradiction that there is a protocol for
{$\sparsegapip^n_{.99q,.9q,.6q}$} with communication complexity less
than {$k_0(q, .05)=\bigOmega{\sqrt{q}}$} that on access to $r \sim_\rho s$ accepts
$\yes$-instances with probability at least $2/3$ and
$\no$-instances with probability at most $1/3$.
This implies that there exist strategies $(\calf = \{f^{(r)}\}_r,\calg =
\{g^{(s)}\}_s)$ such that for every pair of distributions $(\dyes,\dno)$
supported mostly (i.e., with probability $.9$) on $\yes$ and $\no$ instances
respectively, we have $\success_{(\dyes,\dno),\rho}(\calf,\calg) > .1$.
In particular, this holds for the distribution $\dyes'$ as defined in
\autoref{defn:dyesprime} and $\dno$ as defined in \autoref{defn:dno}.

Let $\ext_C$, $\ext_D$ be strategies for \agreement\ as defined in
\autoref{defn:ext}. By \autoref{prop:entropy} we get that
$H_\infty(\ext_C(r)), H_\infty(\ext_D(s)) \geq \log n - O(1)$.
By \autoref{lem:reduction} we also have
$\Pr_{r \sim_\rho s}[\ext_C(r) = \ext_D(s)] \geq \Omega_q(1)$. But this
contradicts \autoref{lem:lb:agreement} which asserts (in particular)
that protocols
extracting $\omega_n(1)$ bits can {only} agree with probability $o_n(1)$.
\end{proof}

\section{Low-influence communication strategies} \label{sec:invariance}
The following theorem states that the expected inner product between two multidimensional Boolean functions without common low-degree influential variables when applied to correlated random strings, is well approximated by the expected inner product of two related functions, this time applied to similarly correlated Gaussians. As, per \autoref{ssec:strategy:preliminaries}, the former quantity captures the behavior of communication protocols, this invariance principle enables one to transfer the study to the (more manageable) Gaussian setting. {(For convenience, in this section we switch to the equivalent view of Boolean functions as being defined on $\{+1,-1\}^n$).}

We denote by $N_{p_1,p_2,\theta}$ the distribution on ${\{+1,-1\}}\times{\{+1,-1\}}$ such that the marginals of $(x,y)\sim N_{p_1,p_2,\theta}$ have expectations respectively $p_1$ and $p_2$, and correlation $\theta$ (see \autoref{def:corr:pair:p1p2eta} for an explicit definition).
\begin{restatable}{theorem}{thminvariance}\label{thm:invariance:principle:bifunction}
  Fix any two parameters $p_1,p_2\in (-1,1)$. For all $\eps\in(0,1]$, $\ell\in\N$, $\theta_0\in[0,1)$ {and {closed} convex sets $K_1,K_2 \subseteq [0,1]^\ell$}, there exist $n_0\in\N$, $d\in\N$ and $\tau\in(0,1)$ such that the following holds. For all $n\geq n_0$, there exist mappings
  \begin{align*}
    T_1&\colon \{ f\colon{\{+1,-1\}}^n\to {K_1} \} \to \{ F\colon\R^n\to {K_1} \} \\
    T_2&\colon \{ g\colon{\{+1,-1\}}^n\to {K_2} \} \to \{ G\colon\R^n\to {K_2} \}
  \end{align*}
such that for all $\theta\in[-\theta_0,\theta_0]$, if $f,g$ satisfy
    \begin{equation}\label{eq:invariance:max:min:inf:condition}
      \max_{i\in[n]}  \min \left( \max_{j\in[\ell] } \infldeg[i]{d}{f_j}, \max_{j\in[\ell] }\infldeg[i]{d}{g_j}\right) \leq \tau
    \end{equation}
  then, for $F=T_1(f)$ and $G = T_2(g)$, we have
  \begin{equation}\label{eq:invariance:conclusion}
    \abs{ \shortexpect_{(x,y)\sim N^{\otimes n}}[ \dotprod{f(x)}{g(y)} ] - \shortexpect_{(X,Y)\sim \mathcal{G}^{\otimes n}} [ \dotprod{F(X)}{G(Y)} ] } \leq \eps.
  \end{equation}
  where $N=N_{p_1,p_2,\theta}$ and $\mathcal{G}$ is the Gaussian distribution which matches the first and second-order moments of $N$, i.e. $\expect{x_i}=\expect{X_i}$, $\expect{x_i^2}=\expect{X_i^2}$ and $\expect{x_iy_i}=\expect{X_i Y_i}$. \end{restatable}
\noindent The theorem follows in a straightforward manner from \autoref{lemma:maxmin:maxmax} and \autoref{thm:discrete:to:gaussian:multivalued:twoway}:
\begin{proofof}{\autoref{thm:invariance:principle:bifunction}}
For  $\eps\in(0,1]$, $\ell\in\N$ and $\theta_0\in(0,1)$ as above, let ${\tau_1} \eqdef \tau(\eps/2,{\ell,}\theta_0)$ as in \autoref{thm:discrete:to:gaussian:multivalued:twoway}. Define the operators $T_1,T_2$ as
\[
	T_1 = T_1^{(2)}\circ  T_1^{(1)}, \quad T_2 = T_2^{(2)}\circ  T_2^{(1)}
\]
where $T^{(1)}_1$, {$T^{(1)}_2$} are the operators from \autoref{lemma:maxmin:maxmax} (for $\eps/2$, $\ell$, $\theta_0$ and ${\tau_1}$ as above, which yield the values of $\tau$, $d$ and $n_0$) and $T_1^{(2)}$, $T_2^{(2)}$ are the (non-linear) ones from \autoref{thm:discrete:to:gaussian:multivalued:twoway} (with parameters $\ell$, $\theta_0$ and $\eps/2$). The result follows.
\end{proofof}

The first step towards proving the theorem is to convert the expected inner product of Boolean functions with no shared low-degree influential variables into expected inner product of Boolean functions with no influential variables at all.
\begin{lemma}\label{lemma:maxmin:maxmax}
  Fix any two parameters $p_1,p_2\in (-1,1)$. For all $\eps\in(0,1]$, $\ell\in\N$, $\tau\in(0,1)$, $\theta_0 \in [0,1)$ {and convex sets $K_1,K_2 \subseteq [0,1]^\ell$},
  there exist $n_0\in\N$, $d \in \N$ and $\tau^\prime\in(0,1)$ such that the following holds. For all $n\geq n_0$ there exist  operators
\begin{align*}
   T^{(1)}_1&\colon \{ f\colon{\{+1,-1\}}^n\to{K_1} \} \to \{ \tilde{f} \colon{\{+1,-1\}}^n\to{K_1} \} \\
   T^{(1)}_2&\colon \{ g\colon{\{+1,-1\}}^n\to{K_2} \} \to \{ \tilde{g} \colon{\{+1,-1\}}^n\to{K_2} \}
\end{align*}
   such that for all $\theta\in[-\theta_0,\theta_0]$, if $f,g$ satisfy
    \begin{equation}\label{eq:lemma:maxmin:maxmax:1}
     	 \max_{i\in[n]}  \min \left( \max_{j\in[\ell] } \infldeg[i]{d}{f_j}, \max_{j\in[\ell] }\infldeg[i]{d}{g_j}\right) \leq \tau^\prime
    \end{equation}
  then, for $\tilde{f}  = T_1^{(1)}{(f)}$ and
  $\tilde{g}  =T_2^{(1)}{(g)}$,
  	\begin{equation}\label{eq:lemma:maxmin:maxmax:2}
      		\max_{i\in[n]}  \max \left( \max_{j\in[\ell] } \infl[i]{\tilde{f}_j}, \max_{j\in[\ell] }\infl[i]{\tilde{g}_j}\right) \leq \tau
    	\end{equation}
 and
 	\begin{equation}\label{eq:maxmin:maxmax:conclusion}
   		 \abs{ \shortexpect_{(x,y)\sim N^{\otimes n}} \dotprod{f(x)}{g(y)} - \shortexpect_{(x,y)\sim N^{\otimes n}} \dotprod{\tilde{f}(x)}{\tilde{g}(y)} } \leq \eps.
 	 \end{equation}
 where $N=N_{p_1,p_2,\theta}$.
\end{lemma}
\begin{proof}
The proof uses Lemmas 6.1 and 6.7 in \cite{MO:GAFA:10} applied to each pair of functions $(f_i,g_i)$, for $i \in
[\ell]$ applied with parameter $\theta_0$ and $\eps/\ell$; using when applying the first lemma the fact that the correlation of these $N_{p_1,p_2,\theta}$ is bounded away from 1.
{The operators given in Lemmas 6.1 and 6.7 in \cite{MO:GAFA:10} are simple averaging operators (averaging the value of $f$
over some neighborhood of $x$ to get its new value at $x$) and by the convexity of $K_1$ we have that the averaged
value remains in $K_1$. Similarly for $g$ and $K_2$.}
We omit the details.
\end{proof}

The last ingredient needed is the actual invariance principle, which will take us from the Boolean, low-influence setting to the Gaussian one.
\begin{restatable}{theorem}{thminvariancetoprove}\label{thm:discrete:to:gaussian:multivalued:twoway}
  Fix any two parameters $p_1,p_2\in(-1,1)$. For all $\eps\in(0,1]$, $\ell\in\N$, $\theta_0 \in [0,1)$,
  and {{closed} convex sets $K_1,K_2 \subseteq [0,1]^\ell$} there exist $\tau > 0$ and mappings
  \begin{align*}
    T^{(2)}_1&\colon \{ f\colon{\{+1,-1\}}^n\to {K_1} \} \to \{ F\colon\R^n\to {K_1} \} \\
    T^{(2)}_2&\colon \{ g\colon{\{+1,-1\}}^n\to {K_2} \} \to \{ G\colon\R^n\to {K_2}\}
  \end{align*}
  such that for all $\theta\in[-\theta_0,\theta_0]$, if $f\colon{\{+1,-1\}}^n\to {K_1}$ and $g\colon{\{+1,-1\}}^n\to {K_2}$ satisfy
  \[
      \max_{i\in[n]} \max\left( \max_{j\in[\ell]} \infl[i]{f_j}, \max_{j\in[\ell]}\infl[i]{g_j} \right) \leq \tau
  \]
  then for $F=T^{(2)}_1(f)\colon\R^n \to{K_1}$ and $G=T^{(2)}_2(g)\colon\R^n \to{K_2}$
  \[
    \abs{ \shortexpect_{(x,y)\sim N^{\otimes n}}\left[ \dotprod{f(x)}{g(y)} \right] - \shortexpect_{(X,Y)\sim \mathcal{G}^{\otimes n}}\left[ \dotprod{F(X)}{G(Y)} \right] } \leq \eps \ ,
  \]
  where $N=N_{p_1,p_2,\theta}$ and $\mathcal{G}$ is the Gaussian distribution which matches the first and second-order moments of $N$.
\end{restatable}
\begin{proof}
Deferred to~\appendixref{sec:appendix:invariance}.
\end{proof}

\subsection{Lower bound for Gaussian Inner Product}We now deduce a lower bound on $k$, the communication complexity of the strategies captured by the range
of $f$ and $g$, needed to achieve sizeable advantage in distinguishing between $\xi$-correlated and uncorrelated Gaussian inputs. {Hereafter, $\mathcal{G}_\rho$ denotes the bivariate normal Gaussian distribution with correlation $\rho$.}
\begin{lemma}\label{lm:gaussian-gapip-lb:twoway}
Let $\xi \in (0,1/2), \gamma > 0$. There exists a function $k_1(\xi,\gamma) \geq \Omega_\gamma(1/\xi)$ such that for every $n$ the following holds: if there are functions $F\colon\R^n \to K_A^{(k)}$ and $G\colon\R^n \to K_B^{(k)}$ such that
\[ | \E_{(x,y) \sim \mathcal{G}_{\xi}^{\otimes n}} [ \dotprod{ F(x) }{ G(y) } ] - \E_{(x,y) \sim \mathcal{G}_0^{\otimes n}} [ \dotprod{ F(x) }{ G(y) } ] | \ge \gamma\ , \]
then $k \geq k_1(\xi,\gamma)$.
\end{lemma}
We will prove the above theorem by translating the above question to a communication lower bound question.

\noindent$\gausscorr_{\xi,n}$: In this (promise) communication game, Alice holds $x \in \R^n$ and Bob holds $y \in \R^n$ from one of two distributions:
\begin{itemize}
\item $\mu_{\yes}$: each $(x_i,y_i)$ is an independent pair of $\xi$-correlated standard normal variables.
\item $\mu_{\no}$: each $(x_i,y_i)$ is an independent copy of uncorrelated standard normal variables.
\end{itemize}
The goal is for Alice and Bob to communicate with each other, with shared randomness,
and distinguish between the two cases with good advantage.

\medskip Note that if $(X,Y)$ denotes the random variable each pair $(x_i,y_i)$ {is a realization of}, estimating $\expect{XY}$ within accuracy $< \xi/2$ suffices to solve the above problem.
If Alice sends the values of $x_i$ (suitably discretized) for the first $O(1/\xi^2)$ choices of $i$,
then by standard Chebyshev tail bounds Bob can estimate $\expect{XY}$ to the desired accuracy, and so this problem can be solved with $O(1/\xi^2)$ bits of (one-way) communication.  We now show that $\Omega(1/\xi)$ is a lower bound.

\begin{lemma}
\label{lm:gaussian-corr-lb:twoway}
Let $\xi  \in (0,1/2)$ and $n$ be sufficiently large. Suppose there is a $k$-bit communication
protocol for \gausscorr$(\xi,n)$ that distinguishes between $\mu_{\yes}$
and $\mu_{\no}$ with advantage $\gamma > 0$.  Then $k \ge \Omega_\gamma(1/\xi)$.
\end{lemma}
Before we prove the result, note that \autoref{lm:gaussian-gapip-lb:twoway} follows immediately with $k_1(\xi,\gamma) = \Omega_\gamma(1/\xi)$,
since by \autoref{prop:ka:kb} the functions $F\colon\R^n \to K_A^{(k)}$ and $G\colon\R^n \to K_B^{(k)}$
simply correspond to strategies for a $k$-bit two-way communication protocol with acceptance
probability given by $\E_{X,Y}[\ip{F(X)}{G(Y)}]$.

\begin{proof}[Proof of \autoref{lm:gaussian-corr-lb:twoway}] The lower bound is proved by reducing the \textsc{Disjointness} problem (in particular a promise version of it) to the
\gausscorr problem.

Specifically we consider the promise \textsc{Disjointness} problem with parameter
$m$, where Alice gets a vector $u \in \{0,1\}^m$ and Bob gets
$v \in \{0,1\}^m$, such that $\normtwo{u}^2=\normtwo{v}^2=\frac{m}{3}$. 
The \yes-instances satisfy $\ip{u}{v} = 1$ while the \no-instances satisfy $\ip{u}{v} = 0$, where the inner product is
over the reals. {H{\aa}stad and Wigderson~\cite{HW:07}} show that
distinguishing \yes-instances from \no-instances requires $\Omega(m)$ bits
of communication, even with shared randomness.

We reduce $\textsc{Disjointness}_m$ to $\gausscorr$ with $\xi=1/m$ as follows:
Alice and Bob share $mn$ independent standard Gaussians $\setOfSuchThat{G_{ij} }{ i \in [n], j \in [m] }$.
Alice generates $x = (x_1,\ldots,x_n)$ by letting
$x_i = {\sqrt{\frac{3}{m}}} \sum_{j=1}^m u_j \cdot G_{ij}$ and Bob generates
$y  = (y_1,\ldots,y_n)$ by letting
$y_i = {\sqrt{\frac{3}{m}}} \sum_{j=1}^m v_j \cdot G_{ij}$.
It can be verified that $x_i$ and $y_i$ are standard Gaussians\footnote{{Namely, for any $i$ we have $\expect{x_i} = \frac{\sqrt{3}}{\sqrt{m}} \sum_{j=1}^m u_j \underbrace{\expect{G_{ij}}}_{=0} = 0$, and
\[
\expect{x_i^2} = \frac{3}{m} \sum_{j=1}^m\sum_{\ell=1}^m u_j u_\ell \underbrace{\expect{G_{ij}G_{i\ell}}}_{=0\text{ if } j\neq \ell} = \frac{3}{m} \sum_{j=1}^m u_j^2 \expect{G_{ij}^2} = \frac{3}{m}\normtwo{u}^2 =1.
\]}}
with $\expect{x_i y_i} = \frac{{3}}{m} \ip{u}{v}$.
Thus \yes-instances of \textsc{Disjointness} map to \yes-instances of
\gausscorr drawn according to $\mu_\yes$
with $\xi = {3}/m$,
and \no-instances map to \no-instances drawn according to $\mu_{\no}$.
The communication lower bound of $\Omega(m)$ for \textsc{Disjointness}
thus translates to a lower bound of $\Omega(1/\xi)$ for
\gausscorr.
\end{proof}

\subsection{Putting things together and proof of \autoref{thm:singlefg}}
\label{sec:putting-together}
We now combine the results from the previous two sections to prove \autoref{thm:singlefg}.
\begin{proof}[Proof of \autoref{thm:singlefg}]
{Postponing} the precise setting of parameters for now, the main idea behind the proof is the following. Suppose the conclusion of the theorem does not hold and $f,g$ do not have a common influential variable so that
\begin{equation}  \label{eq:singlefg:1}
  \max_{i \in [n]} \min\left\{ {\max_{j\in[2^k]} \infldeg[i]{d}{f_j}}, {\max_{j\in[2^k]} \infldeg[i]{d}{g_j}} \right\}  \leq \tau
\end{equation}
for parameters $d,\tau$ that can be picked with an arbitrary dependence on $q,\eps$.

We now associate the domains of $f$ and $g$ with $\{+1,-1\}^n$ in the natural way by mapping $x \in \{0,1\} \to 2x-1 \in  \{+1,-1\}$. This defines us functions $f^\prime\colon\{+1,-1\}^n \to K_A^{(k)}$ and $g^\prime\colon\{+1,-1\}^n \to K_B^{(k)}$
which satisfy the same conditions on influence as $f$. Further, under this mapping, the distribution $B_Y$ is mapped to $N_\mathcal{Y} \equiv N_{2/q-1,0,1.9/q}$ and $B_N$ is mapped to $N_\mathcal{N} \equiv N_{2/q-1,0,0}$ (for $N_{p_1,p_2,\theta}$ as defined in \autoref{thm:invariance:principle:bifunction}). Let $\mathcal{G_Y}$ and $\mathcal{G_N}$ denote bivariate Gaussian distributions whose first two moments match those of $N_\mathcal{Y}$ and $N_\mathcal{N}$ respectively.

Since the ranges of $f^\prime,g^\prime$ are {closed and} convex (from \autoref{prop:ka:kb}) we get,
by applying \autoref{thm:invariance:principle:bifunction} to functions $f^\prime,g^\prime$ and distributions $N_\mathcal{Y},\mathcal{G_Y}$ and $N_\mathcal{N},\mathcal{G_N}$ respectively, that there exist functions $F\colon\R^n \to K_A^{(k)}$ and $G\colon\R^n \to K_B^{(k)}$ such that
\begin{align}\label{eq:singlefg:2}
    \abs{ \shortexpect_{(x,y)\sim N_\mathcal{Y}^{\otimes n}}[ \dotprod{f^\prime(x)}{g^\prime(y)} ] - \shortexpect_{(X,Y)\sim \mathcal{G_Y}^{\otimes n}} [ \dotprod{F(X)}{G(Y)} ] }
    &\leq \frac{\eps}{3} \\
    \abs{ \shortexpect_{(x,y)\sim N_\mathcal{N}^{\otimes n}}[ \dotprod{f^\prime(x)}{g^\prime(y)} ] - \shortexpect_{(X,Y)\sim \mathcal{G_N}^{\otimes n}} [ \dotprod{F(X)}{G(Y)} ] }
    &\leq \frac{\eps}{3} \ .\nonumber
\end{align}
Combining the above equations with the hypothesis that $\success_{(\dyes,\dno)}(f,g) \geq \eps$, we get
\[
\abs{ \shortexpect_{(X,Y)\sim \mathcal{G_Y}^{\otimes n}}[ \dotprod{F(X)}{G(Y)} ] - \shortexpect_{(X,Y)\sim \mathcal{G_N}^{\otimes n}} [ \dotprod{F(X)}{G(Y)} ] } \geq \frac{\eps}{3} \ .
\]
To finish the argument, we shall appeal to \autoref{lm:gaussian-gapip-lb:twoway}. Let $p = 1/q$ and $\theta = .95 p/\sqrt{p-p^2}={\bigTheta{1/\sqrt{q}}}$. Let $\phi\colon\R \to \R$ be defined by $\phi(z) = 2\sqrt{p-p^2} \cdot z + (2p-1)$. It is easy to check that for $(z,w) \sim \mathcal{G_\theta}$, $(\phi(z), w) \sim \mathcal{G_Y}$ and for $(z,w) \sim \mathcal{G}_0$, $(\phi(z),w) \sim \mathcal{G_N}$. Therefore, if we define ${F^\prime\colon\R^n \to K_A^{(k)}}$ by ${F^\prime}(X) = F(\phi(X_1),\ldots,\phi(X_n))$, then the above equation is equivalent to
\[
  \abs{ \shortexpect_{(X,Y)\sim \mathcal{G_\theta}^{\otimes n}}[ \dotprod{{F^\prime}(X)}{G(Y)} ] - \shortexpect_{(X,Y)\sim \mathcal{G}_0^{\otimes n}} [ \dotprod{{F^\prime}(X)}{G(Y)} ] }
  \geq \frac{\eps}{3} \ .
\]
We can now conclude from \autoref{lm:gaussian-gapip-lb:twoway} that {$k \geq \Omega_\eps(1/\theta) = \Omega_\eps(\sqrt{q})$}. To complete the proof of theorem by a contradiction we set the parameters as follows: choose $d,\tau$ in \autoref{eq:singlefg:1} so as to deduce \autoref{eq:singlefg:2} from \autoref{thm:invariance:principle:bifunction} (with $\eps/3$ playing role of $\eps$)  and set {$k_0 = k_1(\theta,\eps/3)$ for $k_1$} as given by \autoref{lm:gaussian-gapip-lb:twoway}.

\end{proof}

\section{Conclusions}
\label{sec:conclusions}

In this paper we carried out an investigation of the power of imperfectly
shared randomness in the context of communication complexity. There
are two important aspects to the perspective that motivated our work: First,
the notion that in many forms of natural communication, the communicating
parties understand each other (or ``know'' things about each other)
fairly well, but never perfectly. This imperfection in knowledge/understanding
creates an obstacle to many of the known solutions and new solutions have to be
devised, or new techniques need to be developed to understand whether the
obstacles are barriers. Indeed for the positive results described in this paper,
classical
solutions do not work and the solutions that ended up working are even
``provably'' different from classical solutions. (In particular they work hard
to preserve ``low influence'').

However, we also wish to stress a second aspect that makes the problems here
interesting in our view, which is an aspect of scale. Often in communication
complexity our main motivation is to compute functions with sublinear
communication, or prove linear lower bounds. Our work, and natural
communication in general, stresses the setting where inputs are enormous,
and the communication complexity one is considering is tiny. This models
many aspects of natural communication where there is a huge context to any
conversation which is implicit. If this context were known exactly to sender
and receiver, then it would play no significant mathematical role. However in
natural communication this context is not exactly known, and resolving this
imperfection of knowledge before communicating the relevant message would be
impossibly hard.  Such a setting naturally motivates the need to study problems
{of input length $n$, but where any dependence on $n$ in the communication complexity
would be impractical.}

We note that we are not at the end of the road regarding questions of this form: Indeed a
natural extension to communication complexity might be where Alice wishes to
compute $f_A(x,y)$ and Bob wishes to compute $f_B(x,y)$ but Alice does not
know $f_B$ and Bob does not know $f_A$ (or have only approximate knowledge
of these functions). If $x$ and $y$ are $n$-bits strings, $f_A$ and $f_B$
might require $2^n$ bits to describe and this might be the real input size.
There is still a trivial upper bound of $2n$ bits for solving any such
communication problem, but it would be interesting to study when, and what
form of, approximate knowledge of $f_A$ and $f_B$ helps improve over this
trivial bound.

Turning to the specific questions studied in this paper a fair number of
natural questions arise that we have not been able to address in this work.
For instance, we stuck to a specific and simple form of correlation in
the randomness shared by Alice and Bob. One could ask what general forms of
randomness $(r,r')$ are equally powerful. In particular if the distribution
of $(r,r')$ is known to both Alice and Bob, can they convert their
randomness to some form of correlation in the sense used in this paper
(in product form with marginals being uniform)?

In \autoref{sec:agreement} we considered the \agreement\ problem where the
goal was for Alice and Bob to agree perfectly on some random string. What if
their goal is only to generate more correlated bits than they start with?
What is possible here and what are the limits?

In the study of perfectly shared randomness, Newman's Theorem~\cite{Newman:91}
is a simple but powerful tool, showing that $O(\log n)$ bits of randomness
suffice to deal with problems on $n$ bit inputs. When randomness is shared
imperfectly, such a randomness reduction is not obvious. Indeed for the
problem of equality testing, the protocol of \cite{BavarianGI:14}
uses $2^n$ bits
of randomness, and our Gaussian protocol (which can solve this with one-way
communication) uses $\mathrm{poly}(n)$ bits. Do $O(\log n)$ bits of
imperfectly shared randomness suffice for this problem? How about for
general problems?

Finally almost all protocols we give for imperfectly shared randomness lead
to two-sided error. This appears to be an inherent limitation (with some
philosophical implications) but we do not have a proof. It would be nice to
show that one-sided error with imperfectly shared randomness cannot lead to
any benefits beyond that offered by private randomness.

\section*{Acknowledgments}
We thank Brendan Juba for his helpful notes \cite{JubaReport:Invariance} on the invariance principle. We thank Badih Ghazi for pointing an error in a previous version of this paper and other discussions, \newer{and Jayanth Naranambadikalpathy Sadhasivan for pointing out an issue in the proof of~\autoref{thm:compress}}. We thank the anonymous referee for extensive comments and
corrections.

 \bibliographystyle{alpha}
 \bibliography{references,added}
 \clearpage
 \appendix
\section{Proofs from \autoref{sec:invariance}}\label{sec:appendix:invariance}

Our goal in this section is to prove the needed invariance principle, as stated in \autoref{thm:discrete:to:gaussian:multivalued:twoway}, that allows us to pass from a correlated distribution on $\{+1,-1\}^2$ to a two-dimensional Gaussian distribution with matching moments. We first formally define the discrete distribution of interest to us.

\begin{definition}\label{def:corr:pair:p1p2eta}
 For parameters $p_1,p_2,\theta \in [-1,1]$, let the distribution $N_{p_1,p_2,\theta}$ on ${\{+1,-1\}}\times{\{+1,-1\}}$ be defined as follows:\footnote{We assume that the parameters $p_1,p_2,\theta$ are such that each of the probabilities is in $[0,1]$.}
 \[
    (x,y) = \begin{cases}
        (+1,+1) &\text{ with probability } \frac{1+\theta}{4} + \frac{p_1+p_2}{4} \\
        (+1,-1) &\text{ with probability } \frac{1-\theta}{4}+\frac{p_1-p_2}{4}\\
        (-1,+1) &\text{ with probability } \frac{1-\theta}{4} - \frac{p_1-p_2}{4}\\
        (-1,-1) &\text{ with probability } \frac{1+\theta}{4} - \frac{p_1+p_2}{4}
    \end{cases}
 \]
 so that $\expect{x}=p_1$, $\expect{y}=p_2$ and $\expect{xy}=\theta$.
\end{definition}
The proof of \autoref{thm:discrete:to:gaussian:multivalued:twoway} relies on two general ingredients.
The first is that replacing $f$ and $g$ by
their {\em smoothened} versions $T_{1-\eta} f$ and $T_{1-\eta} g$
(obtained by applying the Bonami--Beckner noise operator, defined below) does not change the inner product $\dotprod{f(x)}{g(y)}$ much, due to the fact that the components $(x_j,y_j)$ are sampled independently from a bounded correlation space (namely $N_{p_1,p_2,\theta}$ for $\theta < 1$). The second is a multi-dimensional invariance principle asserting that these smoothened functions behave similarly on
Gaussian inputs that have
matching moments, with respect to Lipschitz test functions. We then apply this to the Lipschitz function which is the inner product of appropriately rounded versions of inputs, thereby yielding $K_1$ and $K_2$ valued functions 
in the Gaussian domain with inner product close to $\dotprod{f(x)}{g(y)}$.

\begin{definition}[Bonami--Beckner $T_{1-\eta}$ operator]
Let $(\Omega,\mu)$ be a finite probability space, and $\eta \in (0,1)$.
For a function $h\colon\Omega^n \to \R$, the function $T_{1-\eta} h$ is defined as $T_{1-\eta} h(x) = \E_{y}[ h(y) ]$, where each coordinate $y_i$ is sampled independently as follows:
\begin{itemize}
  \item with probability $(1-\eta)$ set $y_i = x_i$;  and
  \item with probability $\eta$, pick $y_i \in \Omega$ as a fresh sample according to $\mu$.
\end{itemize}
For a vector-valued function, $T_{1-\eta}$ acts component-wise, i.e., if $f = (f_1,\dots,f_\ell)\colon\Omega^n \to \R^\ell$, we define $T_{1-\eta} f = (T_{1-\eta} f_1,\dots,T_{1-\eta} f_\ell)$.
\end{definition}

A useful property of the $T_{1-\eta}$ operator for us is that if $h$ has {convex} range {$K \subseteq [0,1]^\ell$ then so does $T_{1-\eta} h$.}
As stated below, the action of $T_{1-\eta}$ has a particularly nice form when a function is expanded in an orthonormal basis, but this will not be important for us.

\begin{fact}
If a function $h\colon\Omega^n \to \R$ has multilinear expansion
$h(x)= \sum_{\multindex} \hat{h}_{\multindex} \prod_{i=1}^n \chi_{\sigma_i}(x_i)$
w.r.t. an orthonormal ensemble $\mathcal{L}=(\chi_0,\dots,\chi_{b-1})$ of random variables over $\Omega$, then the multilinear expansion of $T_{1-\eta}h$ is given by
$\sum_{\multindex} \hat{h}_{\multindex} (1-\eta)^{|\multindex|}  \prod_{i=1}^n \chi_{\sigma_i}(x_i)$.
\end{fact}

We next state the multi-dimensional invariance principle that we rely on. A version similar to the following is stated formally in \cite[Theorem 10.1]{GHMRC:11} (we have renamed some variables to avoid conflict with other uses in this paper) and
it follows from Theorem 3.6 in the work of Isaksson and Mossel~\cite{isaksson-mossel}.

\begin{theorem}
 \label{thm:multidim-invariance}
  Let $(\Omega,\mu)$ be a finite probability space with the least non-zero
  probability of an atom being at least $\alpha \leq 1/2$.  Let $b = \abs{\Omega}$ and let $\mathcal{L} =
  \{\chi_0=1,\chi_1,\chi_2,\ldots,\chi_{b-1}\}$ be a basis for random variables
  over $\Omega$.  Let $\Upsilon =
  \{\xi_0=1,\xi_1,\ldots,\xi_{b-1}\}$ be an ensemble of real-valued Gaussian random variables with first and second moments matching those of the $\chi_i$'s; specifically:
  \begin{align*}
    \E[\chi_i] = \E[\xi_i]  & & \E[\chi_i^2] = \E[\xi_i^2] & & \E[\chi_i
    \chi_j] = \E[\xi_i \xi_j] & & \forall i,j \in \{1,\dots,b-1\}
  \end{align*}
Let $h = (h_1,h_2,\dots,h_t)\colon\Omega^n \to \R^t$ be a vector-valued function such that $\infl[i]{h_\ell}  \leq \tau$ and  $\Var(h_\ell) \leq 1$ for all $i \in [n]$ and $\ell \in [t]$.
For $\eta \in (0,1)$, let $H_\ell$, $\ell=1,2,\dots,t$, be the multilinear polynomial associated with $T_{1-\eta} h_\ell$ with respect to the basis $\mathcal{L}$, as per \autoref{def:rv-to-poly}.

If $\Psi\colon \R^t \rightarrow \R$ is a Lipschitz-continuous function with
Lipschitz constant $\Lambda$ (with respect to the $L_2$-norm), then
\begin{equation}
\label{eq:inv-principle-guarantee}
\biggl|\E \Bigl[\Psi\bigl(H_1(\mathcal{L}^{n}),  \cdots, H_t(\mathcal{L}^n)\bigr)\Bigr] -
    \E \Bigl[\Psi\bigl(H_1(\Upsilon^{n}),  \cdots, H_t(\Upsilon^n)\bigr)\Bigr] \biggr| \leq
    C(t) \cdot \Lambda \cdot \tau^{(\eta/18) \log(1/\alpha)} = o_{\tau}(1)
\end{equation}
for some constant $C(t)$ depending on $t$, where $H_\ell(\mathcal{L}^n)$ and $H_\ell(\Upsilon^n)$, $\ell \in [t]$, denote random variables as in \autoref{def:poly-to-rv}.
\end{theorem}

Armed with the above invariance principle, we now turn to the proof of \autoref{thm:discrete:to:gaussian:multivalued:twoway}, restated below.

\thminvariancetoprove*
\begin{proofof}{\autoref{thm:discrete:to:gaussian:multivalued:twoway}}

Let $\Omega = {\{+1,-1\}} \times {\{+1,-1\}}$ with the measure $N := N_{p_1,p_2,\theta}$. Define the basis $\mathcal{L} = \{\chi_0,\chi_1,\chi_2,\chi_3\}$ of functions on $\Omega$ as:
\begin{itemize}[itemsep=1pt]
\item $\chi_0=1$,
\item $\chi_1((w_1,w_2)) = w_1$ (where $w_1,w_2 \in \{+1,-1\}$),
\item $\chi_2((w_1,w_2)) = w_2$, and
\item $\chi_3((w_1,w_2)) = w_1 w_2$.
\end{itemize}

We will apply the above invariance principle
\autoref{thm:multidim-invariance} with $t=2\ell$, $h_j = f_j$ and
$h_{\ell+j} = g_j$ for $j \in [\ell]$. We note that while $f_j$, $j
\in [\ell]$ are functions on $\{+1,-1\}^n$, we can view them as 
functions on $\Omega^n$ by simply ignoring the second coordinate. (Thus, for $(x,y) \sim \Omega^n$, $f_j(x,y) = f_j(x)$.) The
multilinear expansion of $f_j$ w.r.t. $\mathcal{L}$ will only involve
$\chi_0$ and $\chi_1$. Similarly, the functions $h_j$'s only depend on
the second coordinate of $\Omega$ and have a multilinear expansion
depending only on $\chi_0,\chi_2$. The function $\Psi\colon \R^{2\ell} \to
\R$ is defined as
\[ \Psi(\mv{a},\mv{b}) = \dotprod{ \round_{{K_1}}(\mv{a}) }{ \round_{{K_2}}(\mv{b}) } \]
for $\mv{a},\mv{b} \in \R^\ell$, where for a {\em {closed} convex} set $K \subset \R^\ell$, $\round_K\colon \R^{\ell} \to \R^{\ell}$ maps a point to its (unique) closest point (in Euclidean distance) in $K$  -- in particular, it is the identity map on $K$. It is easy to see that by the convexity of $K$, $\round_K$
  is a $1$-Lipschitz function,\footnotemark{} and it follows that the function $\Psi$ is $O(\sqrt{\ell})$-Lipschitz. Also, since $T_{1-\eta} f$ is ${K_1}$-valued and $T_{1-\eta} g$ is ${K_2}$-valued on $\{+1,-1\}^n$, the $\round$ functions act as the identity on their images, and hence
\begin{equation}\label{eq:inv-proof-1}
\E\Bigl[ \Psi\bigl(H_1(\mathcal{L}^{n}), \cdots, H_t(\mathcal{L}^n)\bigr)\Bigr] = \E_{(x,y)} \Bigl[ \dotprod{ T_{1-\eta} f(x) }{ T_{1-\eta} g(y) } \Bigr] \ ,
\end{equation}
where $(x,y)$ is distributed according to  $N_{p_1,p_2,\theta}^{\otimes n}$.

\footnotetext{
 To see why, let $a,b$ be two arbitrary points and $a'=\round_K(a)$, $b'=\round_K(b)$. Without loss of generality, we can change the coordinates so that $a'=(0,\dots,0)$ and $b'=(c,0,\dots,0)$ for some $c > 0$: by convexity, the segment $[a'b']$ lies within $K$.  Now, by virtue of $a'$ (resp. $b'$) being the closest point to $a$ (resp. $b$), this implies the first coordinate of $a$ must be non-positive and the first coordinate of $b$ must be at least $c$; but this in turn means the distance between $a$ and $b$ is at least $c$.
}

For $j \in [\ell]$, define real-valued functions
$\tilde{F}_j= H_j(\Upsilon^n)$ and $\tilde{G}_j=H_{\ell +j}(\Upsilon^n)$.
Note that as the multilinear expansion of $T_{1-\eta} f_j$ (resp. $T_{1-\eta} h_j$) only involves $\chi_0,\chi_1$ (resp. $\chi_0,\chi_2$), the multilinear expansion of $\tilde{F}_j$ (resp. $\tilde{G}_j$) {only} involves $\xi_0,\xi_1$ (resp. $\xi_0,\xi_2$). As $\xi_0=1$, the functions $\tilde{F}_j$ (resp. $\tilde{G}_j$) are defined on $\R^n$ under a product measure with coordinates distributed as Gaussians with mean $p_1$ (resp. mean $p_2$) and second moment $1$.

Let $\tilde{F} = (\tilde{F}_1,\dots,\tilde{F}_\ell)$
and $\tilde{G} = (\tilde{G}_1,\dots,\tilde{G}_\ell)$, and finally let
$F\colon \R^n \to {K_1}$ be $F(X) = \round_{{K_1}}(\tilde{F}(X))$ and $G\colon \R^n \to {K_2}$ be $G(Y) = \round_{{K_2}}(\tilde{G}(Y))$.
Note that $F$ (resp. $G$) depends only on $f=(f_1,\dots,f_\ell)$ (resp. $g=(g_1,\dots,g_\ell)$) as required in the statement of \autoref{thm:discrete:to:gaussian:multivalued:twoway}. By construction, it is clear that
\begin{equation}
\label{eq:inv-proof-2}
 \E \Bigl[ \Psi\bigl(H_1(\Upsilon^{n}), \cdots, H_t(\Upsilon^n)\bigr)\Bigr] = \E_{(X,Y)} \Bigl[  \dotprod{ F(X) }{ G(Y) }  \Bigr] \ ,
\end{equation}
for $(X,Y) \sim (\xi_1,\xi_2)^{\otimes n} = \mathcal{G}^{\otimes n}$ where $\mathcal{G}$ is the Gaussian distribution which matches the first and second moments of $N=N_{p_1,p_2,\theta}$.

\noindent Combining \eqref{eq:inv-proof-1} and \eqref{eq:inv-proof-2} with the guarantee \eqref{eq:inv-principle-guarantee} of \autoref{thm:multidim-invariance}, we get that
\begin{equation}
\label{eq:inv-proof-3}
\abs{ \E_{(x,y) \sim N^{\otimes n}} \bigl[ \dotprod{ T_{1-\eta} f(x) }{ T_{1-\eta} g(y) } \bigr] - \E_{(X,Y) \sim \mathcal{G}^{\otimes n}} \bigl[  \dotprod{ F(X) }{ G(Y) }  \bigr] }  \le \eps/2
\end{equation}
for $\tau > 0$ chosen small enough (as a function of $\eps,\ell,p_1,p_2,\theta_0$ and $\eta$). We are almost done, except that we would like to be close to the inner product $\dotprod{f(x)}{g(y)}$ of the original functions, and we have the noised versions in \eqref{eq:inv-proof-3} above. However, as the correlation of the space $N_{p_1,p_2,\theta}$ is bounded away from $1$, applying Lemma 6.1 of \cite{MO:GAFA:10} implies that for small enough $\eta > 0$ (as a function of $\eps,\ell,p_1,p_2,\theta_0$),
\[
\abs{ \E_{(x,y) \sim N^{\otimes n}} \bigl[ \dotprod{T_{1-\eta} f(x)}{T_{1-\eta} g(y)} \bigr] - \E_{(x,y) \sim N^{\otimes n}} \bigl[ \dotprod{f(x)}{g(y)} \bigr] } \le \eps/2  \ .
\]
Combining this with \eqref{eq:inv-proof-3}, the proof of \autoref{thm:discrete:to:gaussian:multivalued:twoway} is complete.
\end{proofof}

\end{document}